\documentclass[12pt,3p,preprint,numbers,authoryear]{elsarticle}

\usepackage{graphicx}
\usepackage{tabu}
\usepackage{amsmath}
\usepackage{amssymb}
\usepackage{subfig,multirow,float,hhline,hyperref}
\usepackage{color,soul}
\usepackage{afterpage}
\usepackage{natbib}
\usepackage{amsthm}

\usepackage{tikz}
\usetikzlibrary{shapes.geometric, arrows}

\usepackage{float}
\restylefloat{table}

\theoremstyle{plain}
\newtheorem{theorem}{Theorem}

\newtheorem{lemma}{Lemma}

\theoremstyle{definition} 
\newtheorem{definition}{Definition}

\theoremstyle{definition}
\newtheorem{remark}{Remark}

\theoremstyle{remark}
\newtheorem*{rationale}{Rationale behind the proof}

\makeatletter
\def\ps@pprintTitle{%
   \let\@oddhead\@empty
   \let\@evenhead\@empty
   \let\@oddfoot\@empty
   \let\@evenfoot\@oddfoot
}
\makeatother

\begin{document}

\thispagestyle{empty}

\begin{minipage}[t][5cm][t]{1.\textwidth}

{\Large

This is a preprint of the paper:

~

\noindent Daniel Lerch-Hostalot, David Meg\'{\i}as, ``Unsupervised steganalysis based on artificial training sets'', \emph{Engineering Applications of Artificial Intelligence}, Volume 50, April 2016, Pages 45-59. ISSN: 0952-1976. \url{http://dx.doi.org/10.1016/j.engappai.2015.12.013}.

}

\end{minipage}

\newpage
\setcounter{page}{1}

\begin{frontmatter}

\title{Unsupervised Steganalysis Based on Artificial Training Sets}

\author[Barcelona]{Daniel Lerch-Hostalot}\ead{dlerch@uoc.edu} 
\author[Barcelona]{David Meg\'{\i}as}\ead{dmegias@uoc.edu}

\address[Barcelona]{Estudis d'Inform\`{a}tica Multim\`{e}dia i Telecomunicaci\'{o}, Internet Interdisciplinary Institute (IN3), Universitat Oberta de Catalunya, Av. Carl Friedrich Gauss, 5, 08660 Castelldefels, Catalonia, Spain.} 

\begin{keyword}     
Unsupervised steganalysis, Cover source mismatch, Machine learning
\end{keyword}

\begin{abstract}
In this paper,  an unsupervised steganalysis method that combines artificial training sets and supervised classification is proposed. \textcolor{black}{We provide a formal framework for unsupervised classification of stego and cover images in the typical situation of targeted steganalysis (i.e., for a known algorithm and approximate embedding bit rate).} We also present a complete set of experiments using \textcolor{black}{1)} eight different image databases, \textcolor{black}{2)} image features based on Rich Models, and \textcolor{black}{3)} three different embedding algorithms: \textcolor{black}{Least Significant Bit (LSB) matching, Highly undetectable steganography (HUGO) and Wavelet Obtained Weights (WOW).} We show that the experimental results outperform \textcolor{black}{previous} methods based on Rich Models in the majority of the tested cases. At the same time, the proposed approach bypasses the problem of Cover Source Mismatch --when the embedding algorithm and bit rate are known--, since it removes the need of a training database when we have a large enough testing set. Furthermore, we provide a generic proof of the proposed framework in the machine learning context. Hence, the results of this paper could be extended to other classification problems similar to steganalysis.
\end{abstract}

\end{frontmatter}

\section{Introduction}
Data hiding is a collection of techniques to embed secret data into digital media. These techniques can be used in many different application scenarios, such as secret communications, copyright protection or authentication of digital contents, among others. \textcolor{black}{Images are the most common carriers for data hiding because of their widespread use in the Internet.}

\textcolor{black}{Within data hiding, steganography is a major branch whose goal is to secretly communicate data, making it undetectable for an attacker.}  On the other hand, \textcolor{black}{steganalysis is another branch whose goal is to detect messages previously hidden using steganography.}

Many of the image steganalysis methods in the state of the art \citep{Pevny:2009,Fridrich:2012} use feature-based  steganalysis and machine learning classification. \textcolor{black}{In order to} apply this methodology, the steganalyst needs to extract a set of features from a training data set and train a classifier. Then, the classifier is tested using a testing data set and, if the results are satisfactory, the classifier is considered successful.

This approach is widely adopted in classification tasks. The classifier is trained with a specific data set and, consequently, its classification capabilities \textcolor{black}{usually} decrease as the testing data set differs from the training data. As a result, this methodology is not fully effective when used in real-world scenarios. 

The data set obtained after feature extraction depends on many factors, such as the steganographic algorithm used for hiding  data into the cover source, the algorithm used for feature extraction or the properties of the cover source in different aspects \textcolor{black}{(e.g. size, noise and hardware used for acquisition).} \textcolor{black}{If similar cover source is used,} the feature extraction process provides data sets with similar representation \textcolor{black}{and, therefore,} the machine learning tools work properly and the classification results are satisfactory. \textcolor{black}{However, if different cover source is used,} the data sets obtained by feature extraction are \textcolor{black}{also different,} producing a degradation of the classification results. Machine learning \citep{Mitchell:1997, Bishop:2006} literature refers to this problem as \emph{domain adaptation}, \textcolor{black}{whereas the term used to refer to this situation in steganalysis is
\emph{cover source mismatch (CSM)}. This constitutes} an important open problem in the field \citep{Ker:2013}, which was initially reported in  \citep{Cancelli:2008}.

Several approaches to deal with the CSM problem have been proposed in the recent years. In the BOSS competition \citep{BOSS}, the BOSSrank database (\textcolor{black}{which suffers from} CSM) had to be used as a testing set. Some participants of the competition tried to include the testing set images in the training set \citep{Kurugollu:2011,Fridrich:2011}. This idea was called ``training on a contaminated database'' \citep{Fridrich:2011}. This approach consists in applying denoising algorithms to estimate the cover sources of the testing set and using these estimated covers to generate new stego samples, by embedding new information into them. After that, these new estimated cover and stego samples are included in the training set. 

In 2012, a solution based on training a classifier with a huge variety of images was proposed \citep{Lubenko:2012}. \textcolor{black}{This} approach consists in applying machine learning \textcolor{black}{to} millions of images. \textcolor{black}{Due to the high time and memory requests, this step is performed using on-line classifiers.} Later on, in 2013, the use of rich features in universal steganalysis was analyzed \citep{Pevny:2013}. Since rich features are not sensitive enough for their application in universal steganalysis, the authors apply linear projections informed by embedding methods and an anomaly detector. This approach tries to make these projections sensitive to stego content and, at the same time, insensitive to cover variation.

In 2014, different methods to deal with CSM were presented. \cite{Ker:2014} show the possibility of centering features when there is a shift in the cover sources, by subtracting an estimated centroid of the cover features. \citeauthor{Ker:2014} \textcolor{black}{also use} weighted ensemble methods to deal with situations in which the features are moving in different directions after embedding. \cite{Fridrich:2014} present three different strategies to deal with CSM. The first one consists in training with a mixture of different cover sources. The second approach uses different classifiers trained on different sources and, in the testing step, the testing set is classified using the closest source. The third strategy is similar to the second one, but testing each image separately using the closest source. In 2014, another approach, based on Ensemble Classifiers with Feature Selection (EC-FS) \citep{Chaumont:2012}, was proposed by \cite{Pasquet:2014}. In this new method, \citeauthor{Pasquet:2014} use the EC-FS classifier with the \emph{Islet approach}, a pre-processing step that consists in organizing images in clusters and assigning a steganalyzer to each cluster. Using this technique, a classifier can manage the diversity of the images more easily, after learning with a set of close feature vectors (in each cluster, the distance between the feature vectors is relatively small). This \textcolor{black}{allows reducing} the number of required images from millions  to a few thousands. 

In this paper we present a new approach based on bypassing the CSM problem rather than addressing it. \textcolor{black}{The proposed} technique consists in creating  an ``artificial'' training set from the testing set. This artificial training set is formed by applying the targeted steganographic algorithm to the testing data (the data set $A$) twice. If the testing set $A$ is formed by stego and cover images, a first application of the steganographic algorithm results in a ``transformed'' set $B$ with ``double stego'' and stego images. The second application of the steganographic algorithm produces a ``double transformed'' set $C$ that includes ``triple stego'' and ``double stego'' images. We show how the sets $A$ and $C$ can be used as artificial training data to finally classify the set $B$ into stego and ``double stego'' images. Since there is a bijection between the elements of $A$ and $B$, this is equivalent to the classification of the images in $A$ as cover or stego.

\textcolor{black}{The idea behind the proposal is} that part of the images that we want to classify --the cover images-- can be transformed into images that belong to the other class that we want to classify: the class of stego images. This fact is exploited to create an artificial training set that is used \textcolor{black}{to find} a boundary between classes with \textcolor{black}{remarkable accuracy.} This classification technique can thus have a relevant impact in the way in which steganalysis is usually approached, since it \textcolor{black}{allows classifying} the images without a real training set, which constitutes the direct cause of the CSM problem. 

\textcolor{black}{In this paper, we provide results for three different steganographic methods, namely, Least Significant Bit (LSB) matching (LSB matching) \citep{Mielikain:2006}, Highly undetectable steganography (HUGO) \citep{Pevny:2010b} and Wavelet Obtained Weights (WOW) \citep{Holub:2012}. Nevertheless, the proposed method is general and can be applied to any steganographic system, such as the more recent methods suggested by \cite{Karakis:2015}.}

The rest of this paper is organized as follows. Section \ref{sec:ATS_method} presents the proposed method, which is formalised in Section \ref{sec:TheoreticalAnalysis}. Section \ref{sec:Experimental} presents the experimental results obtained using the proposed method for \textcolor{black}{eight} different image databases in a cross-domain environment. Finally, Section \ref{sec:Conclusion} summarizes the conclusions of this work and suggests some directions for further research.

\section{Proposed Method}
\label{sec:ATS_method}

This section presents a description of the method proposed in the paper.

\begin{figure}[ht!]
  \subfloat[Original samples: data set $A$]{\label{fig:MEa}
  \ifpdf \includegraphics[width=0.47\textwidth] {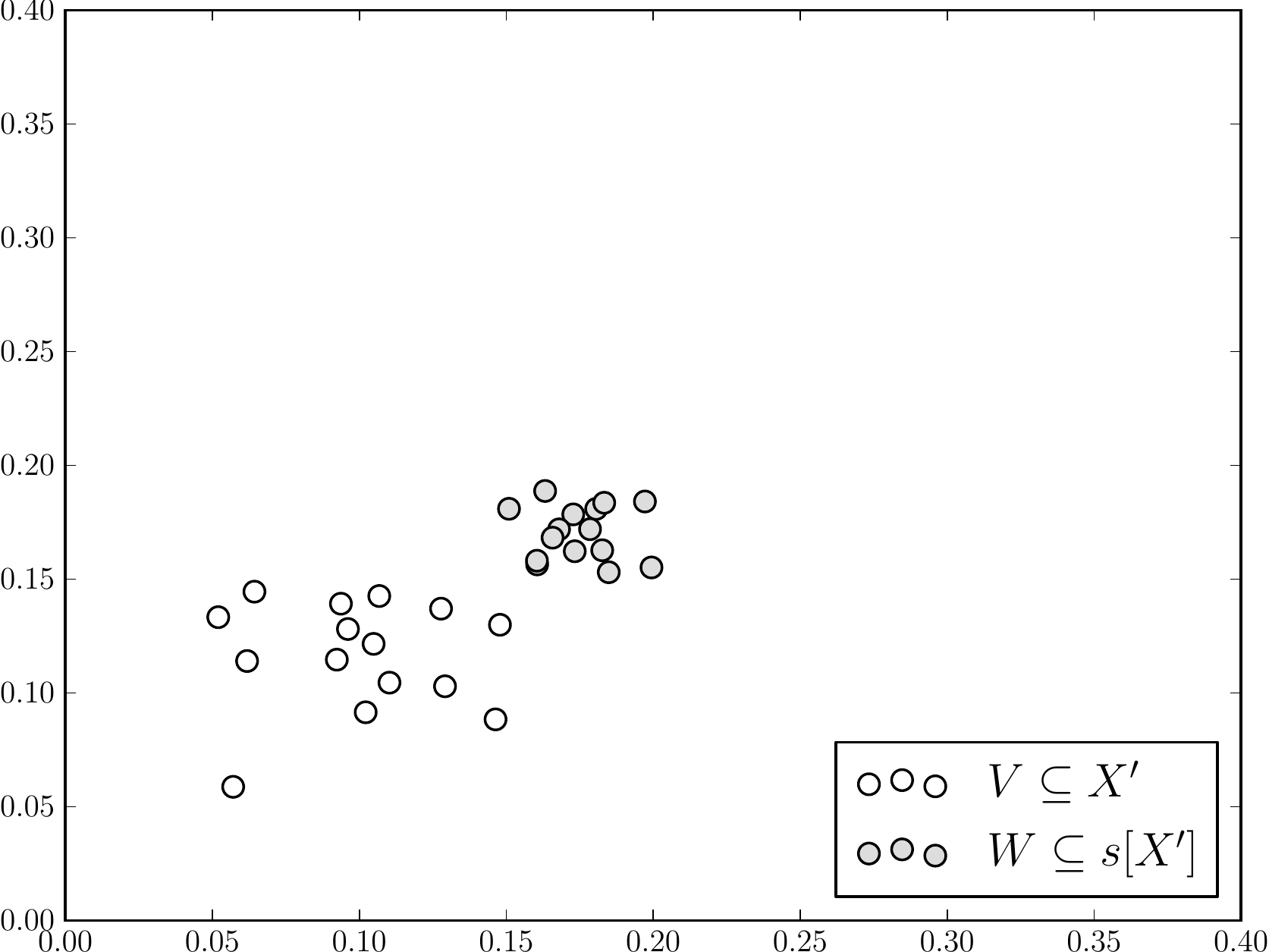}
  \else \includegraphics[width=0.47\textwidth] {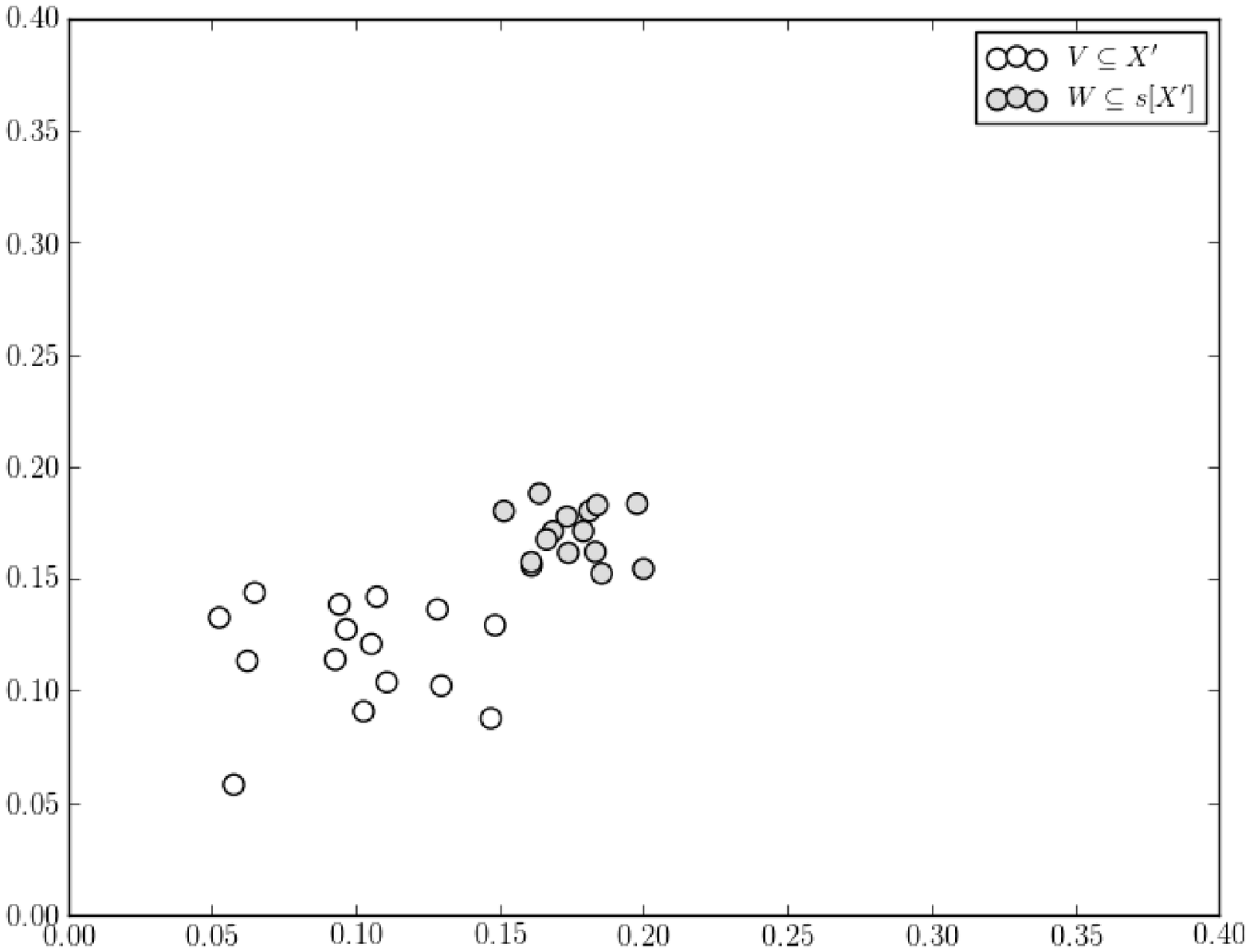}
  \fi} \quad
  \subfloat[``Transformed'' samples: data set $B$]{\label{fig:MEb}
  \ifpdf \includegraphics[width=0.47\textwidth]{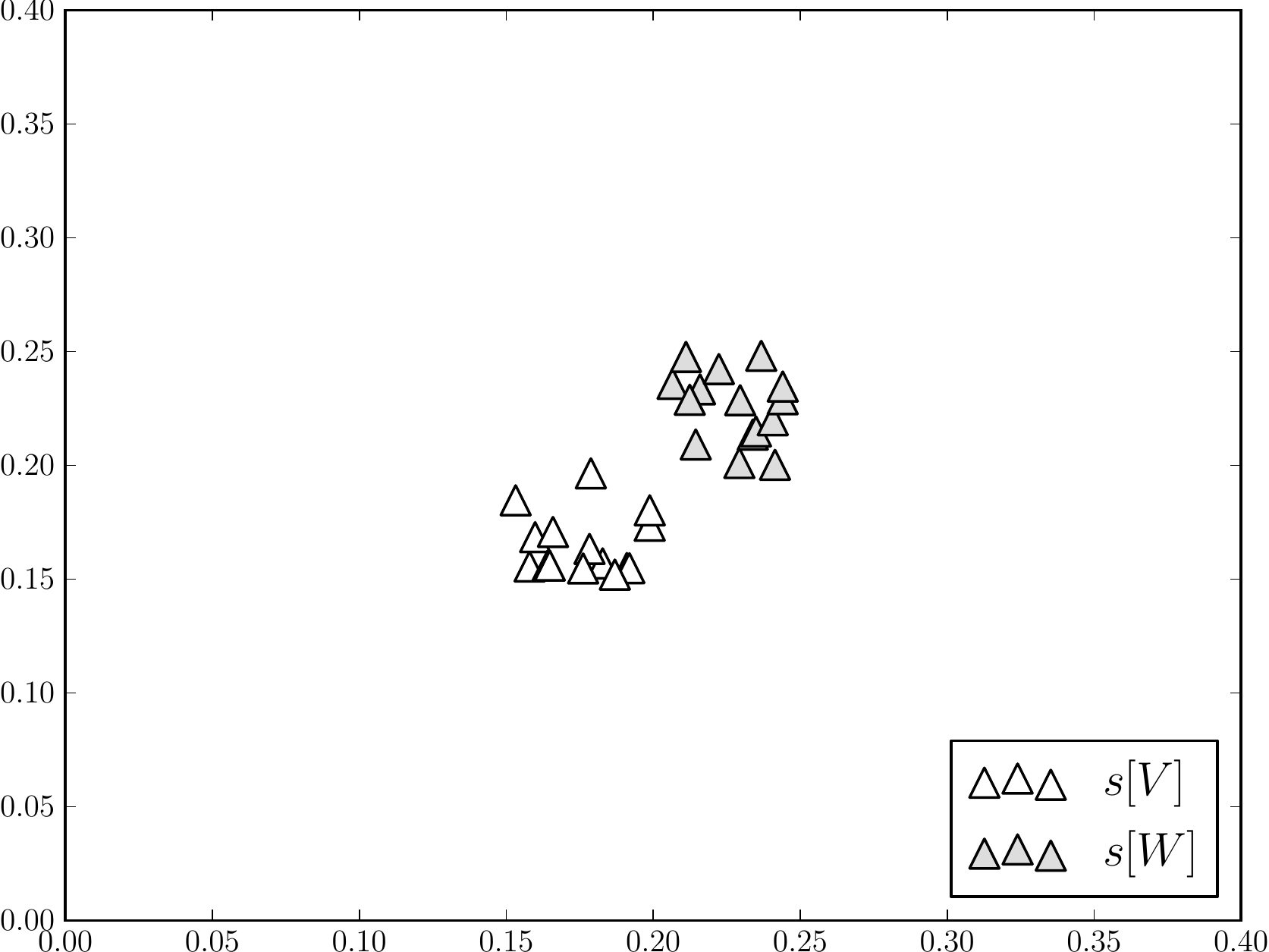}
  \else \includegraphics[width=0.47\textwidth] {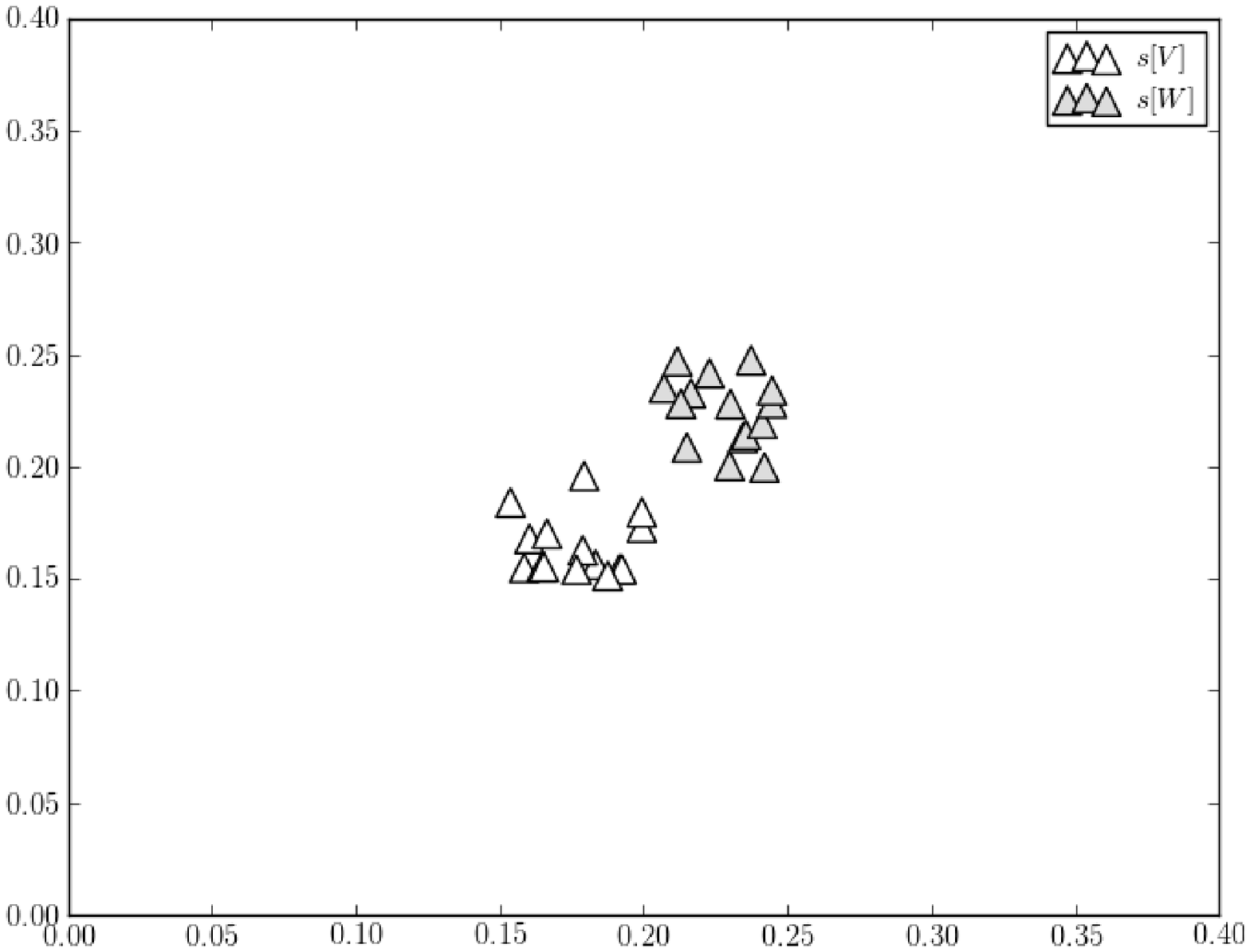}
  \fi} \\
  \subfloat[``Double transformed'' samples: data set $C$]{\label{fig:MEc}
  \ifpdf \includegraphics[width=0.47\textwidth]{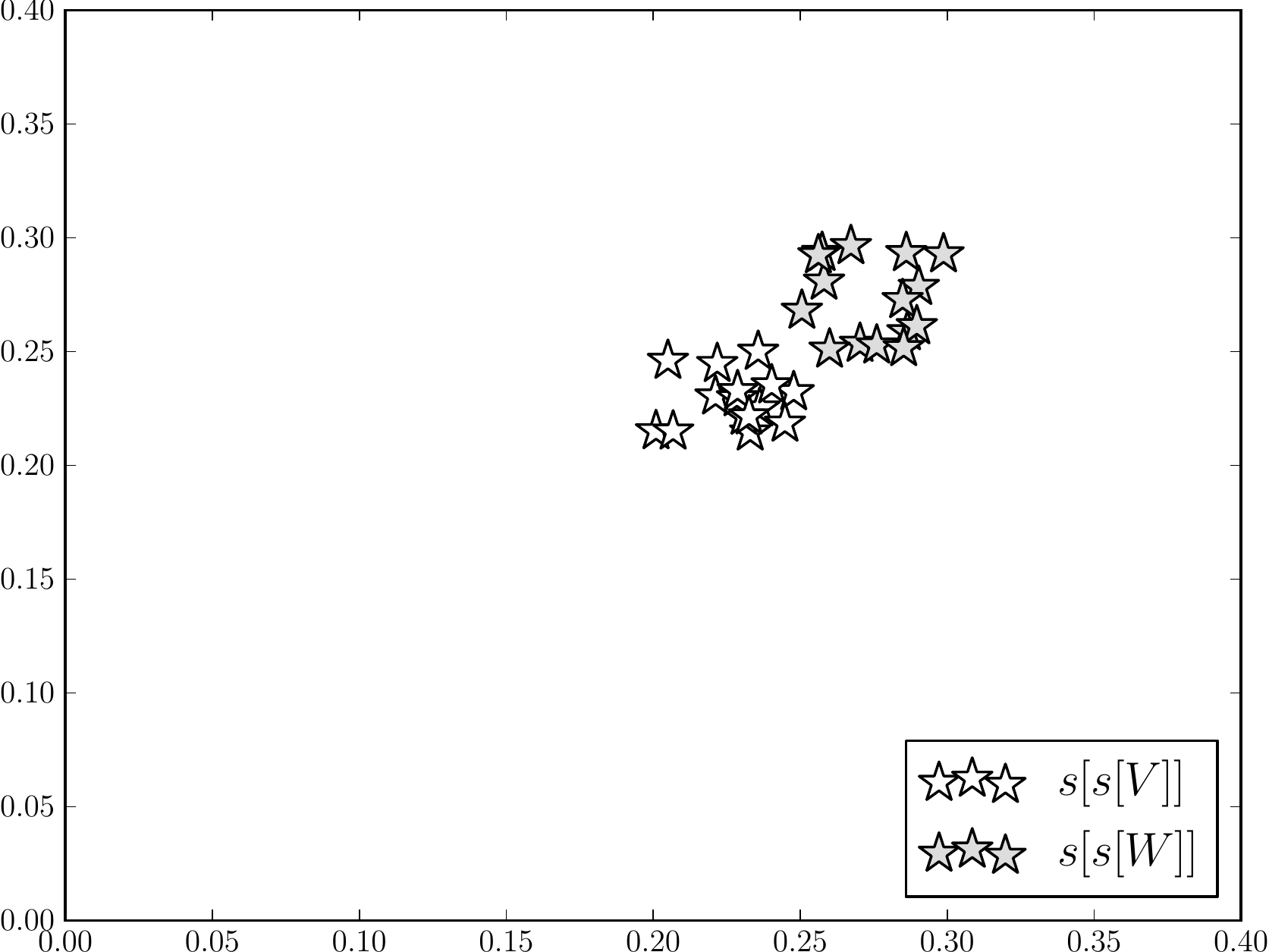}
  \else \includegraphics[width=0.47\textwidth]{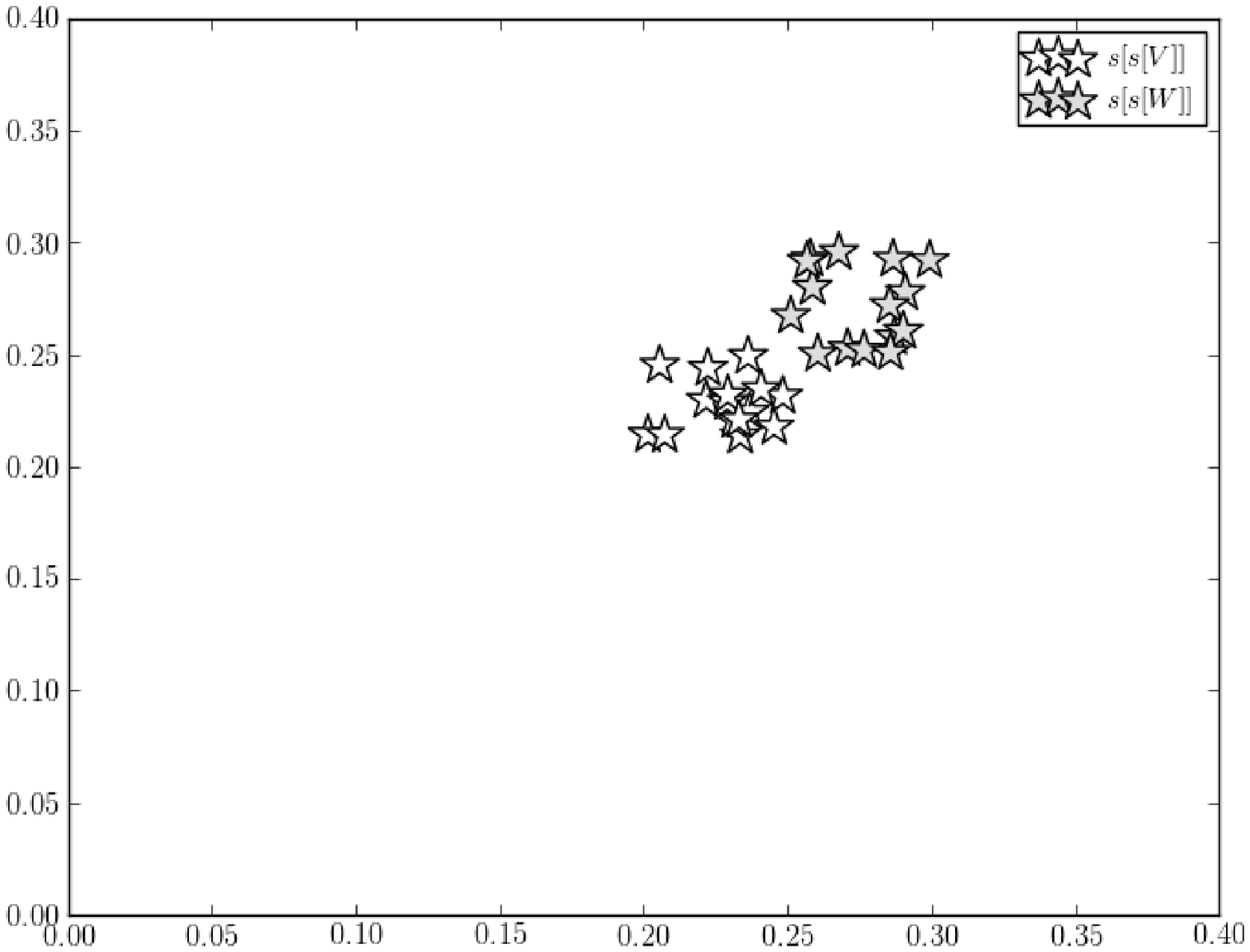}
  \fi} \quad
  \subfloat[All data sets]{\label{fig:MEd}
  \ifpdf \includegraphics[width=0.47\textwidth]{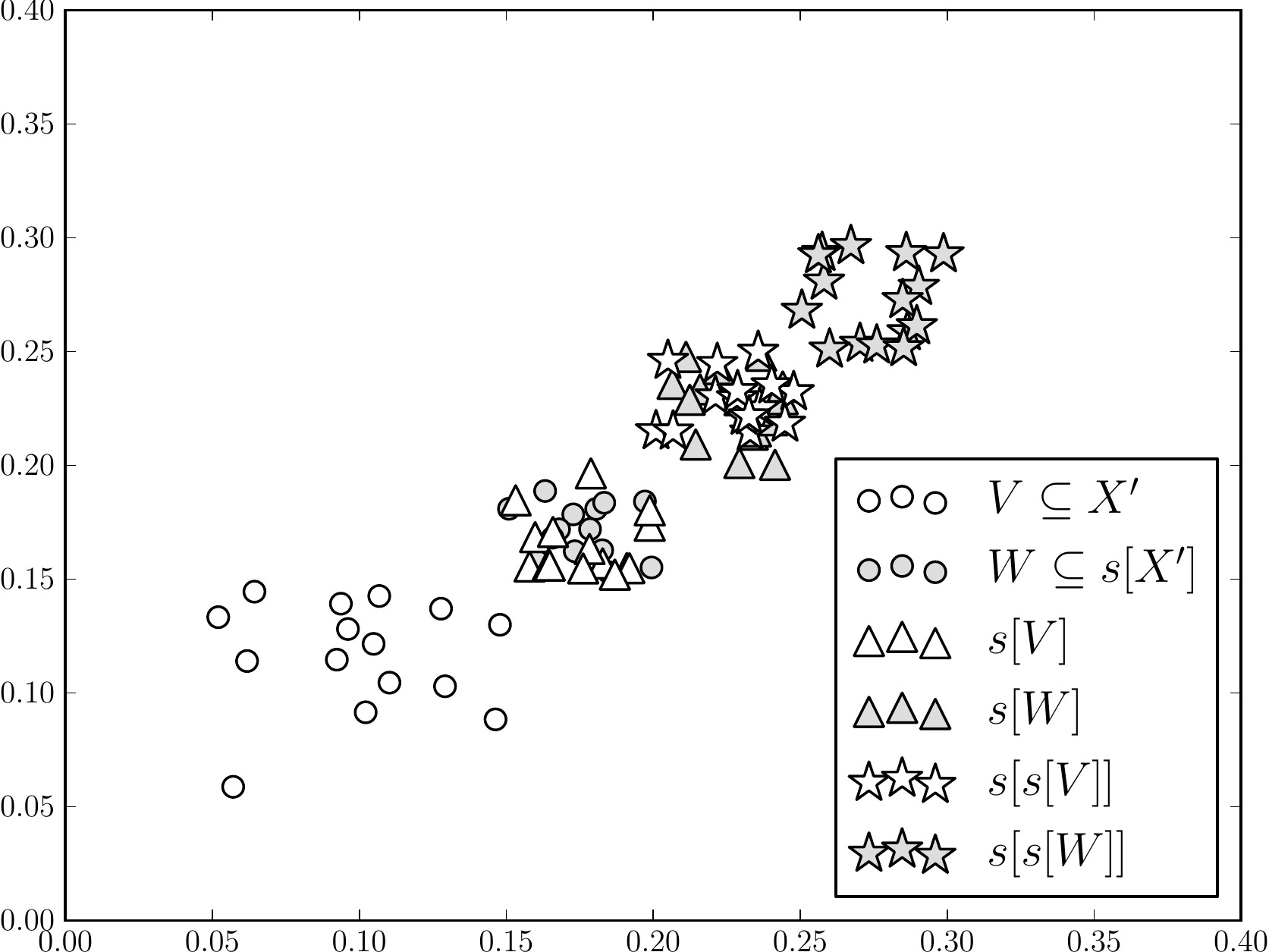}
  \else \includegraphics[width=0.47\textwidth]{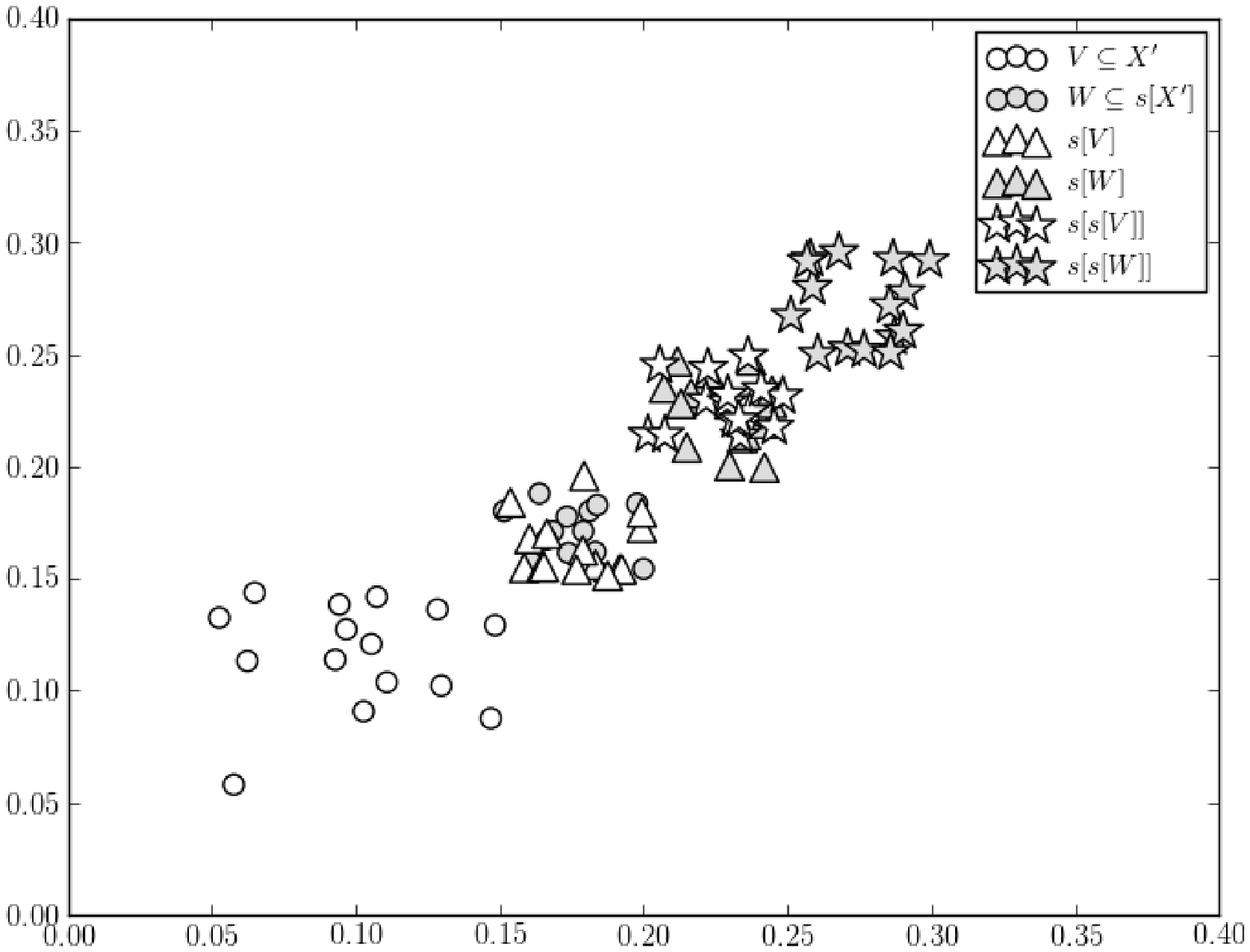}
  \fi} \\
  \subfloat[Learning boundary $A/C$]{\label{fig:MEe}
  \ifpdf \includegraphics[width=0.47\textwidth]{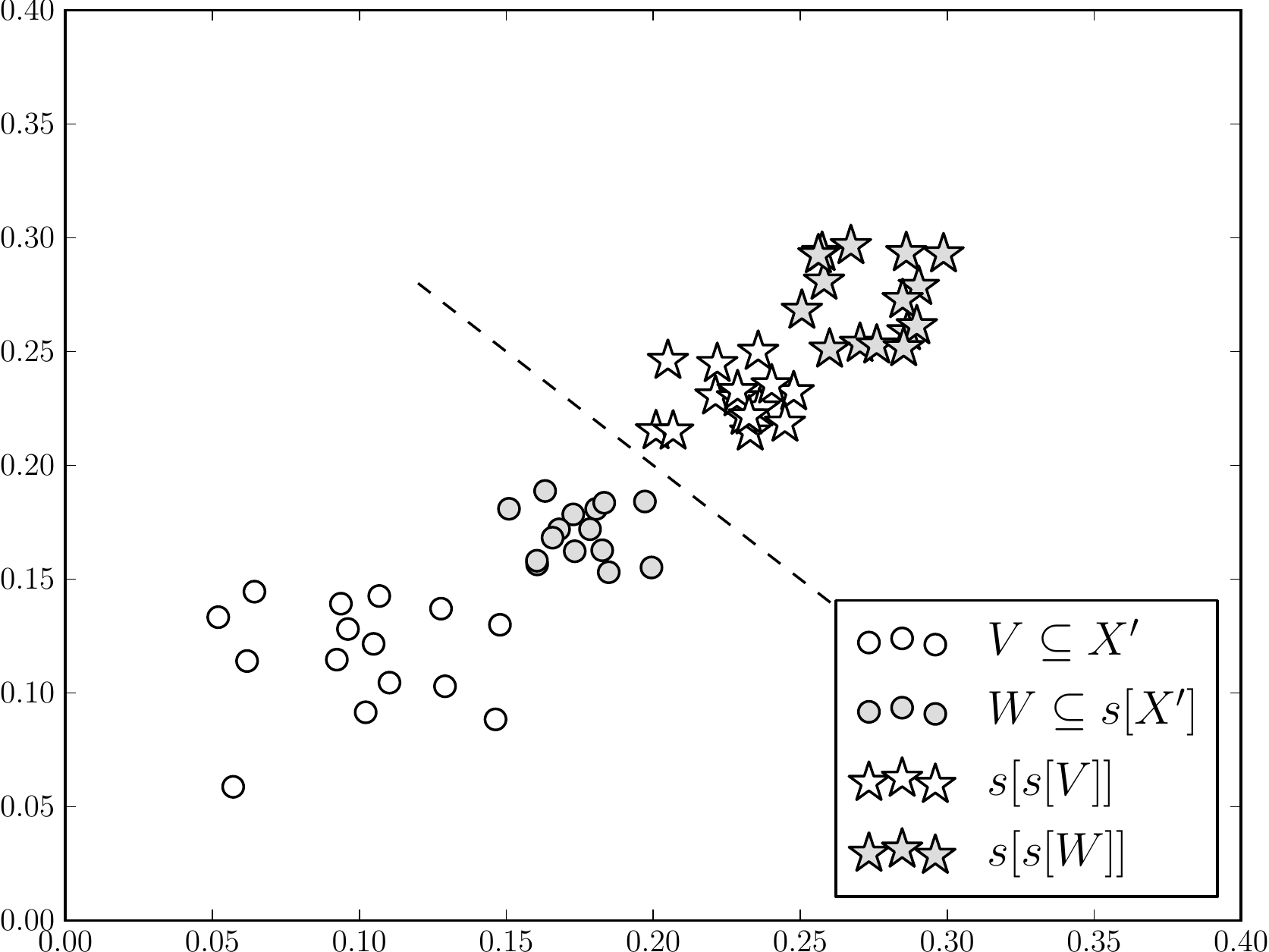}
  \else \includegraphics[width=0.47\textwidth]{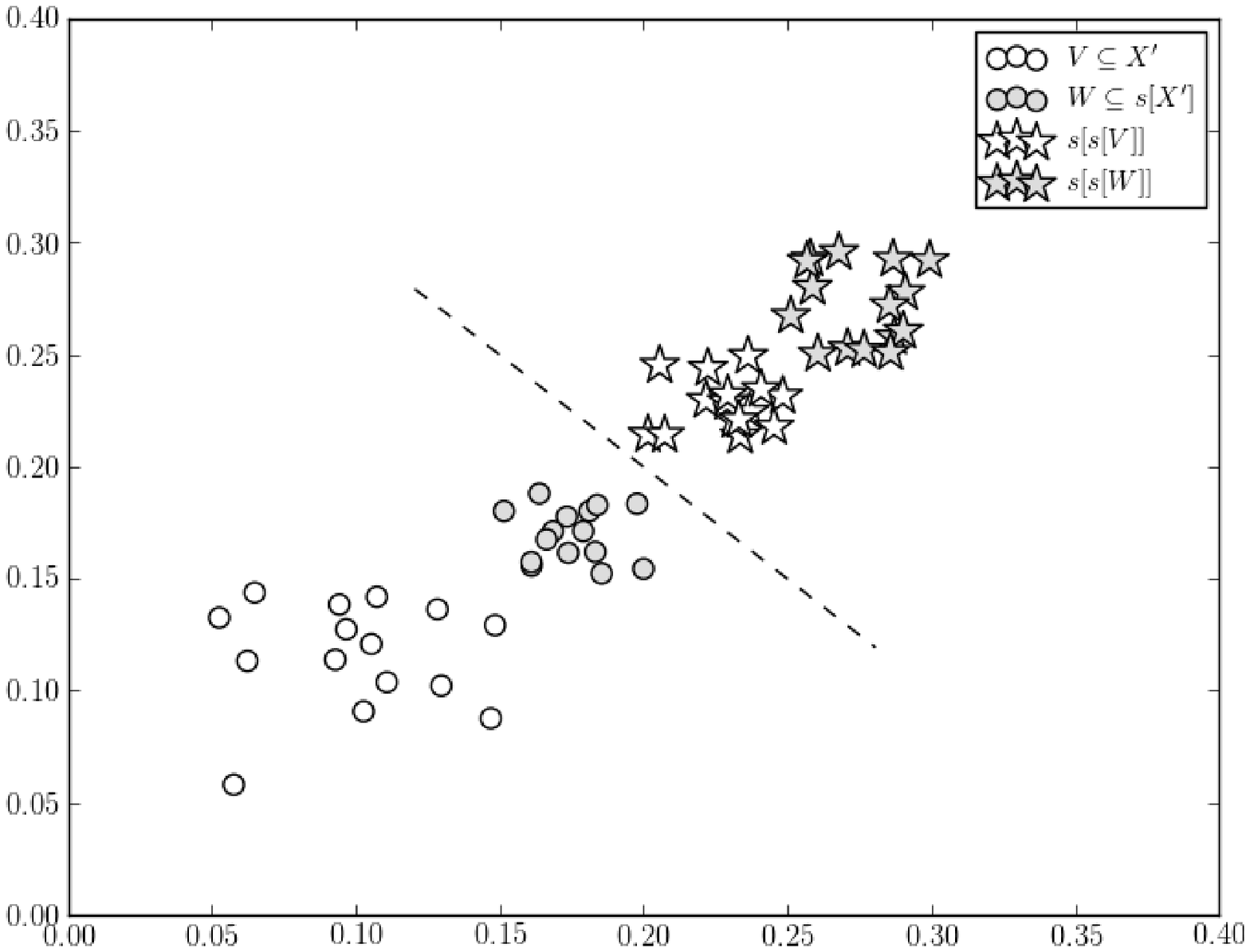}
  \fi} \quad
  \subfloat[Classification of the data set $B$]{\label{fig:MEf}
  \ifpdf \includegraphics[width=0.47\textwidth]{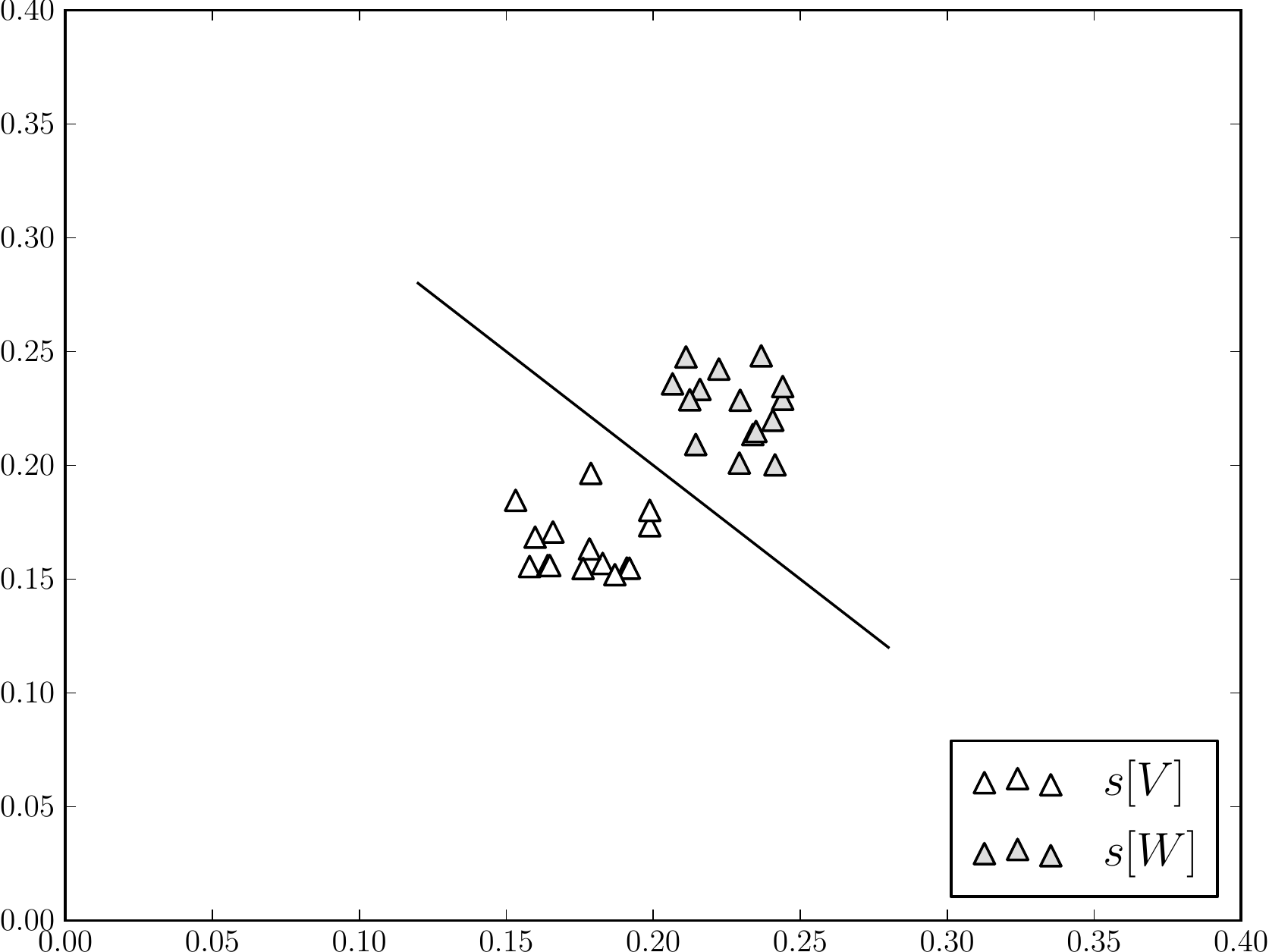}
  \else \includegraphics[width=0.47\textwidth]{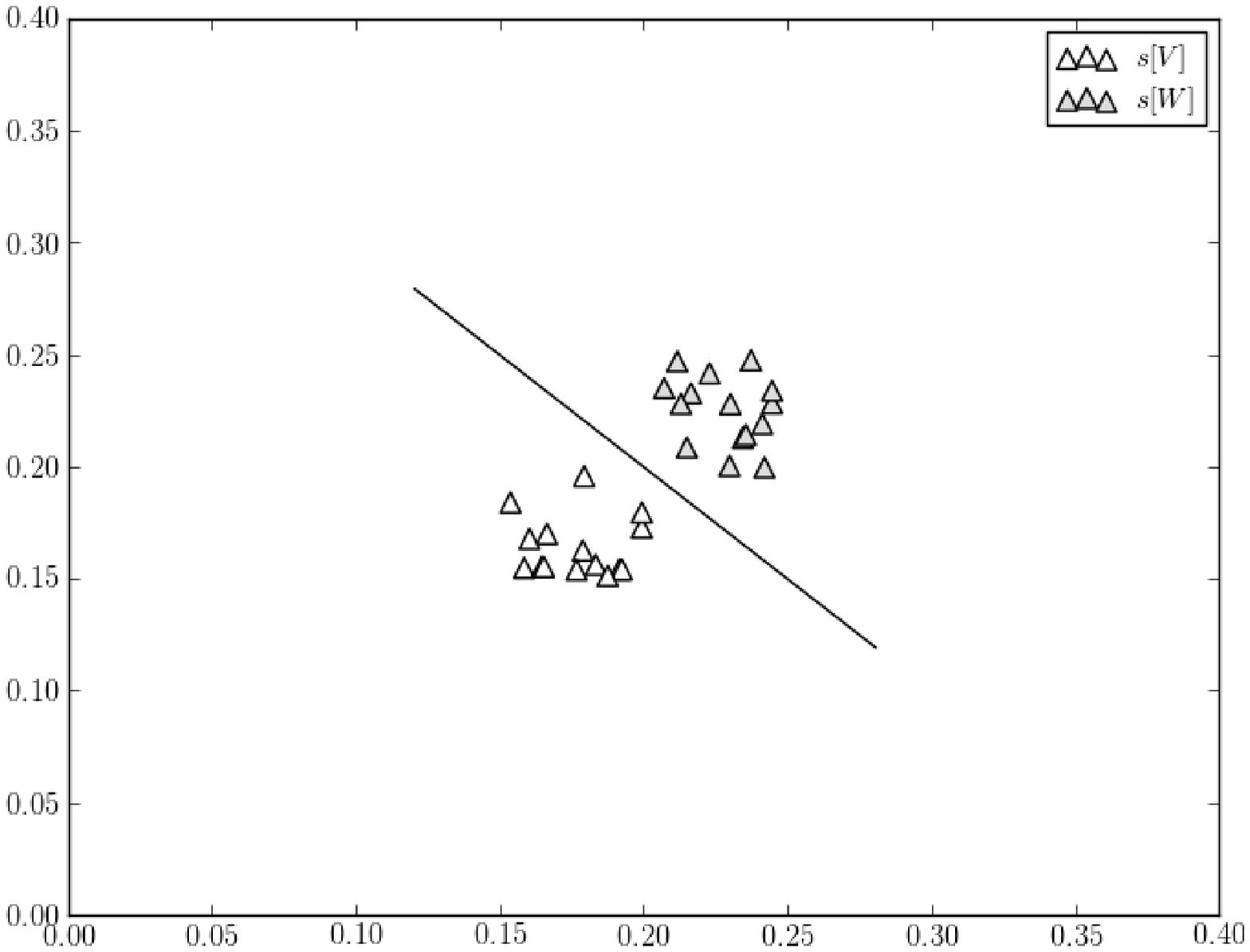}
  \fi} \\
  \caption{Simple graphical representation of the proposed method}
  \label{fig:ATS_method}
\end{figure}

\textcolor{black}{First of all, we} assume that the embedding algorithm and the approximate bit rate used by the steganographer \textcolor{black}{are known}. Using the same algorithm and bit rate, we can perform new embedding operations to all the images in the testing database. Let $A$ be the testing data set; \textcolor{black}{$B$ the transformed set, obtained after embedding data in all the images of $A$; and $C$ the double transformed set, obtained after embedding data in all the images of $B$.} As a result, $A$ contains cover and stego images, $B$ contains stego and ``double stego'' images, and $C$ contains ``double stego'' and ``triple stego'' images. The interesting fact of this situation is that, if we create an artificial training set formed by $A$ and $C$, we can train a classifier to learn the boundary between $A$ and $C$. In principle, the same boundary can be used to classify the transformed set $B$ into stego  and ``double stego'' images. Furthermore, the existing bijection between the elements of $B$ and $A$ makes it possible to relate each stego image of $B$ with a cover  image in $A$, and each ``double stego'' image in $B$ with a stego image in $A$. Hence, classifying $B$ as stego or ''double stego'' images is equivalent to classifying $A$ as cover or stego images. The bijection between the elements of $A$ and $B$ can be recorded in order to complete the classification of the original testing set $A$. 

A simple graphical representation of this approach is shown in Fig.\ref{fig:ATS_method} to \textcolor{black}{illustrate} the rationale behind the proposed algorithm. In Fig.\ref{fig:MEa}, we can see the set $A$, with the cover and stego samples depicted as white and gray circles, respectively. \textcolor{black}{Although the cover and stego images are shown with circles of different color, note that this set is not labeled for classification (since it is the testing data set). In Fig.\ref{fig:MEb}, the set $B$, resulting from the application of the steganographic algorithm to all the images of the $A$, are shown. In this case, 
the cover images of $A$ become stego images of $B$ (white triangles), whereas the stego images of $A$ become ``double stego'' images of $B$ (gray triangles). Similarly, in Fig.\ref{fig:MEc}, we can see the ``double transformed'' set $C$, which contains ``double stego'' (white stars) and ``triple stego'' images (gray stars). In Fig.\ref{fig:MEd}, all the data sets $A$, $B$ and $C$ are shown together.} 

\textcolor{black}{The boundary between $A$ and $C$ can be found using machine learning with $A \cup C$ as a training set, as shown in Fig.\ref{fig:MEe}. In this step, two different labels must be used, one for the images of $A$ and the other one for those of $C$. Then, this trained classifier can be applied to the set $B$ as depicted in Fig.\ref{fig:MEf}, where the learnt boundary is used to classify the images of $B$ as stego or ``double stego''. Finally, this result can be applied to separate $A$ into cover and stego images, since stego images in $B$ match cover images in $A$, whereas ``double stego'' images in $B$ match stego images in $A$.} 


\begin{figure}[ht!]
\centering 
  \ifpdf \includegraphics[width=0.95\textwidth] {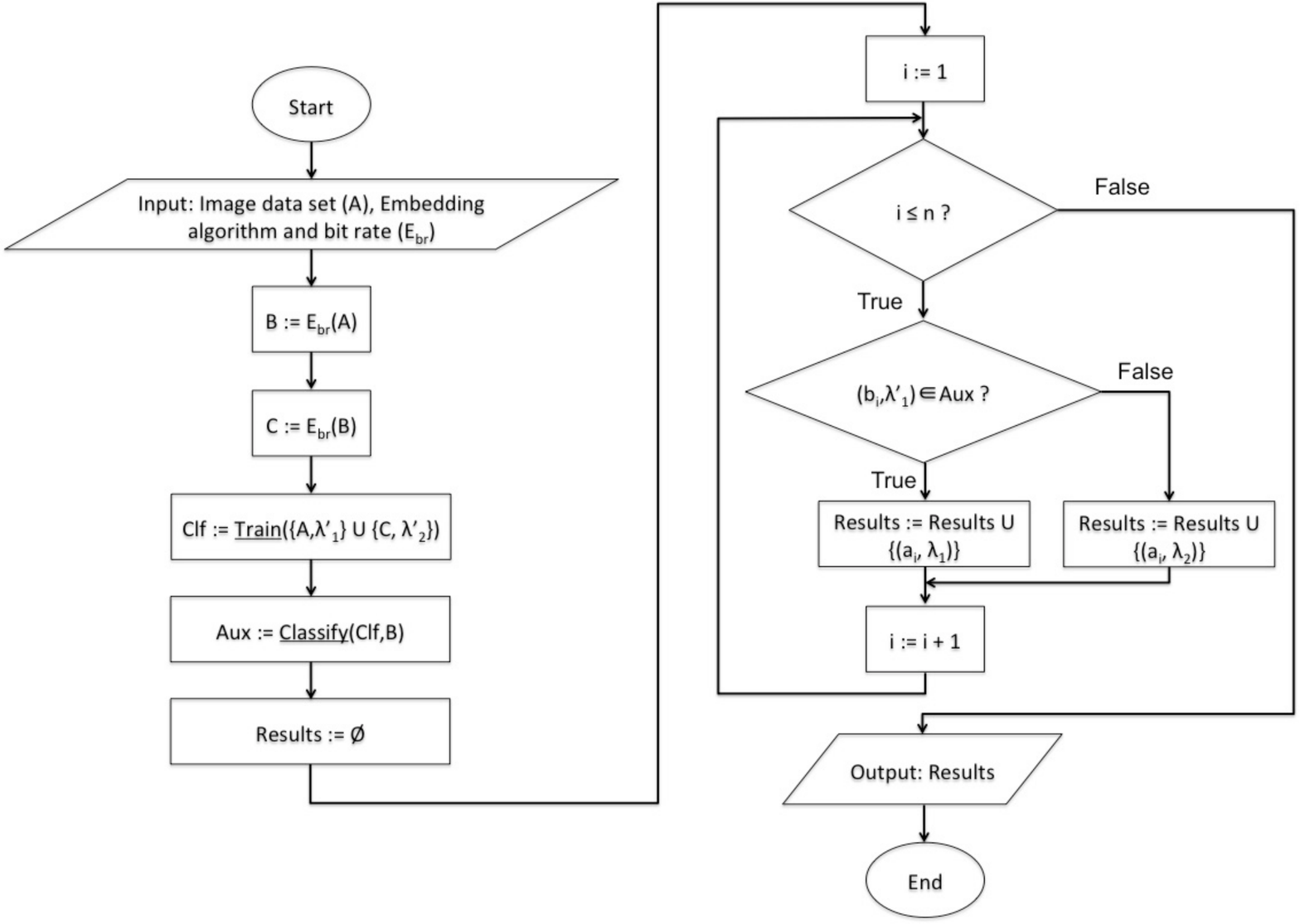}
  \else \includegraphics[width=0.95\textwidth] {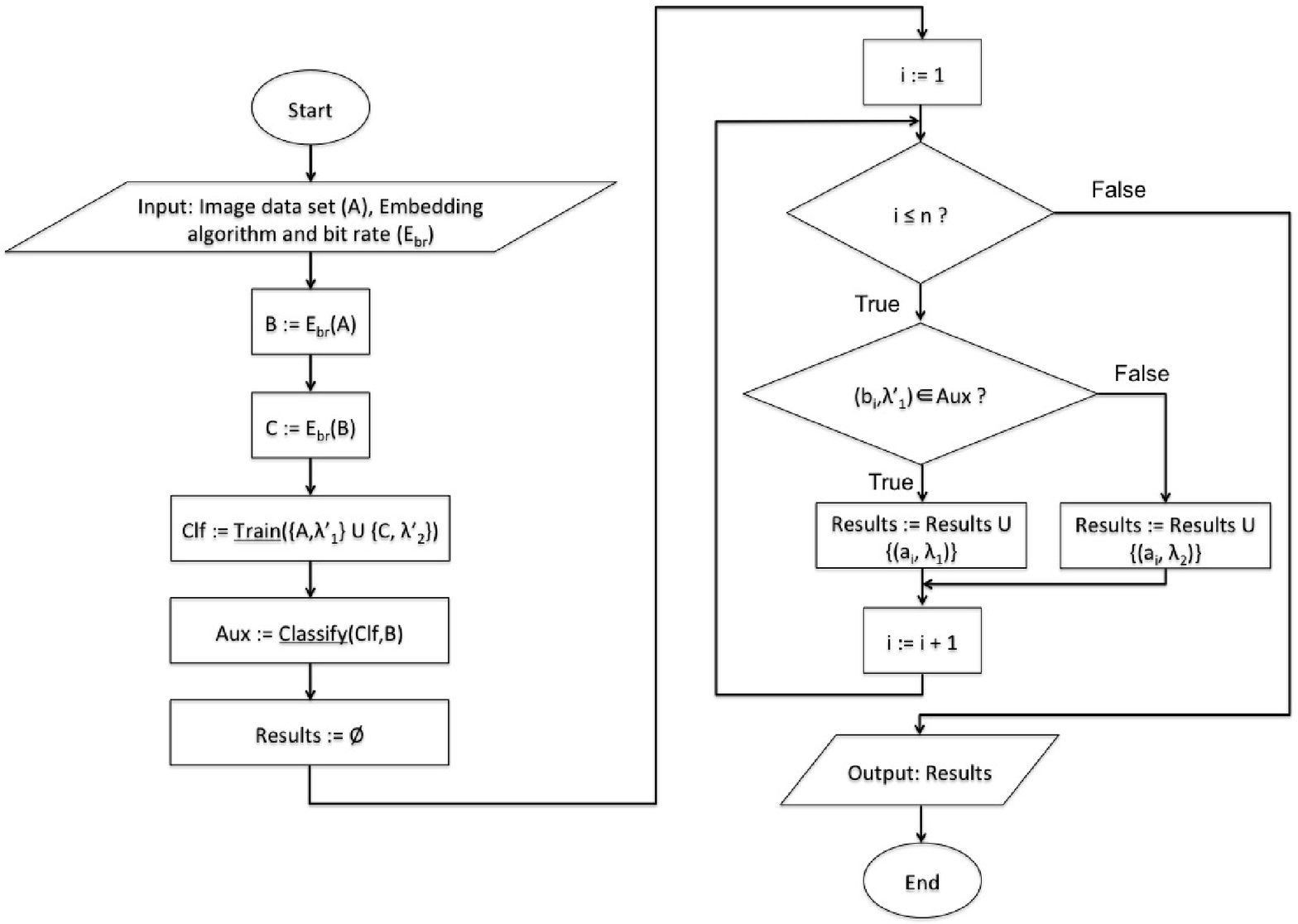}
  \fi
\caption{Flowchart of the proposed method}
\label{fig:flowchart}
\end{figure}

\textcolor{black}{
A flowchart of the proposed algorithm is shown in Fig.\ref{fig:flowchart}. For the sake of notational simplicity, $(A,\lambda'_1)$ and $(B,\lambda'_2)$ stand for the Cartesian products $A\times \{\lambda'_1\}$ and $B\times \{\lambda'_2\}$, respectively. Similarly, a function applied to a set, e.g. $E_{br}(A)$, means that the function is applied to all the elements (images) of that set, and the resulting image set is returned. The call to the function \underline{\em Train} returns a classifier \textit{Clf}. This classifier is then used to separate $B$ into stego ($\lambda'_1$) and ``double stego'' ($\lambda'_2$) images, by calling the function \underline{\em Classify}. The final loop ``translates'' the classification of the images of $B$, as stego ($\lambda'_1$) or ``double stego'' ($\lambda'_2$), into the classification of the images of $A$, as cover ($\lambda_1$) or stego ($\lambda_2$). In this loop, $n=|A|=|B|=|C|$, where $|\cdot|$ stands for the cardinality of a set.}

\textcolor{black}{
It must be taken into account that the notation in the flowchart is deliberately abused, since the training and classification procedures are not carried out directly with the images, but with their feature vectors which must be extracted before.}

A theoretical analysis of the approach is provided in the next section. The theoretical analysis includes a theorem and a proof based on the assumptions taken (which are standard assumptions in the field of targeted steganalysis).

\section{Theoretical Analysis}
\label{sec:TheoreticalAnalysis}
This section is aimed at providing a theoretical framework for the suggested method, by means of definitions, lemmas, a theorem and proofs. 


First of all, we provide a formalisation of learning algorithms with some basic definitions.
Consider an input space $X$ (of samples) and an output space $L=\{\lambda_1,\lambda_2,\dots,\lambda_M\}$ (of labels), and assume that the pairs $(x_i,l_i) \in X \times L$ are random variables i.i.d. according to an unknown distribution $D$, such that each $x_i$ is associated to a unique label $l_i$.

The goal of a classification (learning) algorithm is to find a function $h:X\to L$ that predicts $l_i$ from $x_i$. The function $h$ is an approximation of an unknown \textcolor{black}{labeling} function $t: X\to L$ that relates each value of $X$ with its corresponding unique label in $L$.

There are two types of classification methods:
\begin{itemize}
\item \textit{Supervised classification}: in this case, we use a \textit{training set} of $m$ \textbf{known} samples (pairs) $\{(x_1,l_1),(x_2,l_2),$ $\dots$ $,(x_m,l_m)\}\subset X \times L$ in a training fase to obtain the function $h$. After that, the function $h$ is used to classify a \emph{testing set} of $n$ vectors $\{x'_1,x'_2,\dots,x'_n\}\subset X$, without any label assigned to them.
\item \textit{Unsupervised classification}: in this case, the learning algorithm obtains the function $h$ to classify a testing set without requiring any training set.
\end{itemize}

Now, we provide some definitions required for the proposed method and a formal proof by means of a theorem.

\begin{definition}
We call \emph{splitting function} to a function $s: X' \to s[X']$, where  \linebreak $X', s[X'], s[s[X']], s[s[s[X']]] \subset X$, such that:
\begin{enumerate}
\item When $s(x_i)=x_j$ and $t(x_i)=\lambda_k$, then $t(x_j)=\lambda_{k+1}$ (with $\lambda_k \in L$ for $k<M$), 
\item $X' \cap s[X']=\emptyset$, 
\item $X' \cap s[s[X']]=\emptyset$, and
\item $X' \cap s[s[s[X']]]=\emptyset$.

\end{enumerate}
\label{def:transform}
\end{definition}

\begin{definition}
Given a set $V$ of $n_1$ vectors $\{v_1,v_2,\dots,v_{n_1}\}$ $\subseteq$ $X'$, such that $t(v_i)=\lambda_1$ for all $i=1,2,\dots,n_1$, and $X'\subset X$, and a set $W$ of $n_2$ vectors $\{w_1,w_2,\dots,w_{n_2}\}\subseteq s[X']$, whereby $s[X'] \subset X$, for some \emph{splitting function} $s$ (and hence $X'\cap s[X']=\emptyset$, $X'\cap s[s[X']]=\emptyset$ and $X'\cap s[s[s[X']]]=\emptyset$), we call $A=V \cup W$ an \emph{$s$-partable set}.
\label{def:separable}
\end{definition}

\begin{remark}
Note that the label corresponding to the vectors of $W$, according to Def.\ref{def:transform}, is $\lambda_2$.
\end{remark}

\begin{remark}
The sets $V$ and $W$ are disjoint as per Condition 2 of Def.\ref{def:transform} and, thus, we assume that some feature set exists to classify the elements of $A$ into $V$ and $W$ by using machine learning. In an $s$-partable set, we should be able to approximately classify the set $A=V \cup W$ into the subsets $V$ (with label $\lambda_1$) and $W$ (with label $\lambda_2$), hence the name ``$s$-partable''. This would occur if the splitting function produces some ``measurable difference'' when applied to the original data vectors.
\end{remark}

\begin{lemma} Given an $s$-partable set $A\subseteq X' \cup s[X']$, for some splitting function $s$, then $s[s[X']] \cap A = \emptyset$.
\label{lemma:first}
\end{lemma}

\begin{proof}
\[
\begin{split}
s[s[X']] \cap A &\subseteq s[s[X']] \cap \left(X' \cup s[X'] \right) \\
&= \left(s[s[X']] \cap X'\right) \cup \left(s[s[X']] \cap s[X']\right) \\
&= \emptyset \cup \emptyset \\
&= \emptyset.
\end{split}
\]
Note that $\left(s[s[X']] \cap X'\right)=\emptyset$ due to Condition 3 in Def.\ref{def:transform} and $\left(s[s[X']] \cap s[X']\right)=\emptyset$ due to Condition 2 in Def.\ref{def:transform} (with $Y'=s[X']$).
\end{proof}

\begin{lemma} Given an $s$-partable set $A\subseteq X' \cup s[X']$, for some splitting function $s$, then $s[X'] \cap s[s[A]] = \emptyset$.
\begin{proof}
\[
\begin{split}
s[X'] \cap s[s[A]] &\subseteq s[X'] \cap \left(s[s[X']] \cup s[s[s[X']]] \right) \\
&= \left(s[X'] \cap s[s[X']]\right) \cup \left(s[X'] \cap s[s[s[X']]]\right) \\
&= \emptyset \cup \emptyset \\
&= \emptyset.
\end{split}
\]
Note that $\left(s[X'] \cap s[s[X']]\right)=\emptyset$ due to Condition 2 in Def.\ref{def:transform} (with $Y'=s[X']$) and \linebreak $\left(s[X'] \cap s[s[s[X']]]\right)=\emptyset$ due to Condition 3 in Def.\ref{def:transform} (again, with $Y'=s[X']$). 
\end{proof}
\label{lemma:second}
\end{lemma}
\begin{lemma} Given an $s$-partable set $A\subseteq X' \cup s[X']$, for some splitting function $s$, then $A \cap s[s[A]] = \emptyset$.
\begin{proof}
The proof directly follows from considering 
\[ A \cap s[s[A]] \subseteq (X' \cup s[X']) \cap (s[s[X']] \cup s[s[s[X']]]),\]
and then applying the distributive property of set algebra and Conditions 2, 3 and 4 of Def.\ref{def:transform}.
\label{lemma:third}
\end{proof}
\end{lemma}

\begin{theorem}
\emph{Unsupervised classification} of an \emph{$s$-partable} set $A=\{a_1, a_2, \dots, a_n\}$, formed as per Def.\ref{def:separable} for some splitting function $s$, can be achieved with \emph{supervised classification} of $B=s[A]$ as a \emph{testing set} if the function $s$ is known (or can be approximated) by constructing an artificial training set.

\begin{rationale}
The proof provided below is based on the following principles:
\begin{enumerate}
\item The set $A$ of samples to be classified and the set of ``double transformed'' samples $C=s[s[A]]$ are disjoint as per Lemma \ref{lemma:third}. Hence, we assume that a machine learning algorithm (for some set of features) can be trained to separate $A$ and $C$. Thus, we merge these two sets for training (using different labels for the elements of $A$ and $C$).
\item The set of ``transformed samples'' $B=s[A]$ has elements either in $s[X']$ or in $s[s[X']]$. These two sets are disjoint as per Condition 2 of Def.\ref{def:transform}.
\item The intersection between the subset of the elements of $B$ belonging to $s[X']$ and $C$ is empty, but the intersection of the same subset with $A$ is not. Hence, if we use the trained algorithm to classify $B$, there is a high probability that the elements of $B$ belonging to $s[X']$ will be classified with \textcolor{black}{the label used for the}  
elements of $A$.
\item Similarly, the intersection between the subset of the elements of $B$ belonging to $s[s[X']]$ and $A$ is empty, but the intersection of the same subset with $C$ is not. Hence, the elements of $B$ belonging to $s[s[X']]$ will be classified mainly with the same label as those of $C$.
\item Separating the elements of $B$ into $s[X']$ and $[s[s[X']]$ is equivalent to separating the elements of $A$ into $X'$ and $s[X']$, due to the existing bijection between the elements of $A$ and $B$.
\end{enumerate}
\end{rationale}

\begin{proof}
Since $A$ is an \emph{$s$-partable} set by definition (for the splitting function $s$), some elements $a_i$ belong to the set $X'$ and some other belong to $s[X']$. A classification function must assign the labels $\lambda_1$ and $\lambda_2$ for the elements $a_i$ belonging to $X'$ and $s[X']$, respectively.

Consider two additional sets $B=\{b_1, b_2, \dots, b_n\}=s[A]$, with $b_i=s(a_i)$, and $C=\{c_1, c_2, \dots, c_n\}=s[B]=s[s[A]]$, with $c_i=s(b_i)=s(s(a_i))$. Since $A\subseteq X' \cup s[X']$, then $B\subseteq  s[X'] \cup s[s[X']]$ and $C \subseteq s[s[X']] \cup s[s[s[X']]]$. Now, construct an auxiliary \emph{training set} $T$ as follows: 
\[
T=\{(a_1,\lambda'_1),(a_2,\lambda'_1),\dots,(a_n,\lambda'_1),
(c_1,\lambda'_2),(c_2,\lambda'_2),\dots,(c_n,\lambda'_2)\},
\]
such that $T \subset (A \cup C) \times L'$ for $L'=\{\lambda'_1,\lambda'_2\}$. Note that the label $\lambda'_1$ is used for the elements of $A$ and the label $\lambda'_2$ is used for the elements of $C=s[s[A]]$.

If we use \emph{supervised classification} with the training set $T$, we can obtain a classifying function $h: T \to L'$. This classification is possible since, as proven in Lemma \ref{lemma:third}, the intersection between $A$ and $C=s[s[A]]$ is the empty set. Now, if the function $h$ can classify $T$ into $A$ and $C$, the same function would classify $B$ into $s[X']$ and $s[s[X']]$, since $s[s[X']]\cap A=\emptyset$ (as proven in Lemma \ref{lemma:first}) and $s[X']\cap C=\emptyset$ (as proven in Lemma \ref{lemma:second}). On the other hand, $s[X'] \cap A$ and $s[s[X']] \cap C$ are not empty. Note also that, by Def.\ref{def:transform}, $s[X'] \cap s[s[X']]=\emptyset$.

Thus, using the same function $h$, we can classify each $b_i$ with the label $\lambda'_1$ (if it belongs to $s[X']$) or $\lambda'_2$ (if it belongs to $s[s[X']]$). Therefore, this function also classifies $A$ into $X'$ and $s[X']$, since $b_i=s(a_i)$. If $b_i$ is classified with the label $\lambda'_1$, this means that $b_i=s(a_i)$ belongs to $s[X']$ and $a_i$ belongs to $X'$. Similarly, if $b_i$ is classified with the label $\lambda'_2$, then $b_i=s(a_i)$ belongs to $s[s[X']]$ and $a_i$ belongs to $s[X']$. Hence, the classification of $b_i=s(a_i)$ as $\lambda'_1$ or $\lambda'_2$ is equivalent to the classification of $a_i$ as $\lambda_1$ or $\lambda_2$, respectively. Consequently, we can finally classify $A$ without labeled samples and, by definition, this is \emph{unsupervised classification}.
\end{proof}
\end{theorem}

\begin{remark}
Unsupervised classification of $A$ is thus achieved due to the application of the splitting function $s$ to $A$ (twice), which allows \textcolor{black}{creating} an artificial \emph{training set} and to use \emph{supervised classification}. Hence, if we have a \emph{training set} and we use it to find a function
$h:X\to L$ for classifying a \emph{testing set}, by definition, this
is \emph{supervised classification}.
\end{remark}

\begin{remark}
Note, however, that some degree of misclassification of the elements in $B$ is possible, because the elements in $B$ will not be either in $A$ nor $C$. Since $A=V \cup W$ will usually be a strict subset of $X' \cup s[X']$, the elements $b_i\in s[X']$ will be outside $W$ (and outside $A$), and some of them may be classified incorrectly with the label $\lambda'_2$. Similarly, the elements $b_j\in s[s[X']]$ will be outside $s[s[V]]$ (and outside $C$), and some of them may be classified  incorrectly with the label $\lambda'_1$. This will lead to some degree of error in the unsupervised classification of $B$. The larger the difference introduced by the \emph{splitting function} $s$ is, the least likely the misclassification will become. In any case, it must be taken into account that machine learning classification processes are not error-free.
\end{remark}

\subsection{Application to Steganalysis}
The application of the framework described above to steganalysis can be carried out in the following way:
\begin{itemize}
\item The set $X$ is formed by samples of feature vectors of images (some of them stego and some of them cover).
\item The splitting function $s$ is the steganographic method we want to detect (using approximately the same embedding bit rate). We assume no knowledge about the secret keys of the steganographic scheme.
\item The set $X'$ represents the features of cover images, whereas the set $s[X']$ represents the features of stego images. Note that $X' \cap s[X']=\emptyset$, since an image cannot be cover and stego at the same time.

\item Note that the above condition also implies that successive applications of the embedding method produce disjoint sets. Since the condition applies to any subset $X' \subset X$, we also require, for example, that $s[X'] \cap s[s[X']] = \emptyset$. Some steganographic algorithms may not satisfy this condition. For example, LSB replacement\footnote{\textcolor{black}{Not to be mistaken for LSB matching.}} with an embedding bit rate close to 1 bit per pixel (bpp) may not produce significant differences between $s[X']$ and $s[s[X']]$. Since the proposed method is designed for targeted steganalysis, this condition can be considered known by the steganalyst.
\item  Another problem related to this condition is the fact that the ``splitting'' properties of the embedding process must be fulfilled for some specific feature set, which is used for training the artificial training set and then classifying the ``transformed samples''. For example, the \textcolor{black}{Adaptive Steganography by Oracle (ASO)} embedding algorithm \citep{Kouider:2013} is particularly designed to enhance its undetectability with respect to the features used by Ensemble Classifiers as proposed in \citep{Kodovsky:2012}. Hence, if we use this set of features for classification, it is likely that the suggested approach is not effective for ASO steganography.
\item The set $s[s[X']]$ is formed by the features of ``double stego'' images and, again, we have $X' \cap s[s[X']]=\emptyset$, since an image cannot be cover and ``double stego'' at the same time. The same applies for ``triple stego'' images: $X' \cap s[s[s[X']]]=\emptyset$.
\item The $s$-partable set $A=V \cup W$ is the testing set for the machine learning classification problem. This set is formed by some cover images (belonging to $V \subset X'$) and some stego images (belonging to $W\subset s[X']$).
\item The label $\lambda_1$ corresponds to cover images, whereas $\lambda_2$ refers to stego images.
\item The application of the splitting function $s$ to $V$ obviously produces stego images with no possible intersection with the set of cover images ($X'$). The same thing occurs after different applications of the splitting function, producing ``double stego" images, ``triple stego'' images, and so on.
\item $B=s[A]$ is formed by ``stego'' ($s[V]\subset s[X']$) and ``double stego'' ($s[W]\subset s[s[X']]$) images.
\item $C=s[B]=s[s[A]]$ is formed by ``double stego'' ($s[s[V]]\subset s[s[X']]$) and ``triple stego'' ($s[s[W]]\subset s[s[s[X']]]$) images. 
\end{itemize}
With these definitions, the method described in this section can be applied for the classification of stego and cover images, yielding an unsupervised steganalytic system. 
 
\section{Experimental Results}
\label{sec:Experimental}

This section presents the results obtained with the proposed method and a comparison with existing techniques in the \textcolor{black}{literature to illustrate} the performance of the suggested approach.
 
In order to test the proposed approach, we have selected eight distinct
image databases:
 
\begin{itemize}

\item \textcolor{black}{The BOSS database is the set of training images for the competition Break Our Steganographic System! \citep{BOSS}. This database is formed by 10,000 cover images with a fixed size of $512 \times  512$ pixels, obtained with seven different cameras \citep{Bas:2011}. The cover images  were provided together with a stego set, obtained after embedding the cover images using HUGO steganography with $0.40$ bpp. Hence, the whole training set was formed by 20,000 images. However, after the competition, a weakness was found 
\citep{Kodovsky:2011} in the creation of the stego set, which can be removed embedding with a different threshold ($T = 255$ instead of the default value $T = 90$ that was used in the BOSS challenge). Therefore, we have repaired the database of 20,000 images by replacing the stego set with the 10,000 cover images embedded with HUGO steganography and $0.40$ bpp, using the correct value of the threshold.}

\item \textcolor{black}{The RANK database (often referred to as BOSSrank) is the set of testing images from the competition Break Our Steganographic System! \citep{BOSS}. This database is formed by 1,000 cover images with a fixed size of $512 \times  512$ pixels. These images were first provided in the competition as a testing set and, hence, they were unlabeled and different from those of BOSS. 847 of the RANK images were taken using one of the cameras used to generate the BOSS database, but the remaining 153 images were taken with a camera not included in the BOSS database \citep{Bas:2011}. Thus, these 153 images exhibit the CSM problem with respect to the BOSS database. After the competition, the whole cover set was released. Again, we have repaired the stego set using the correct threshold for HUGO steganography with $0.40$ bpp (i.e., $T=255$). The repaired database is formed by 942 images, 471 of which are cover and 471 stego.}

\item The NRCS database consists of images from the National Resource Conservation 
System \citep{NRCS} with a fixed size of $2100 \times 1500$ pixels.

\item The ESO database consists of images from the European Southern Observatory
\citep{ESO} with variable sizes about $1200 \times 1200$ pixels.

\item The Interactions database consists of images from Interactions.org 
\citep{INTERACTIONS} with variable sizes about $600 \times 400$ pixels.

\item The NOAA database consists of images from the National Oceanic and Atmospheric 
Administration \citep{NOAA} with variable sizes about 
$2000 \times 1500$ pixels.  

\item The Albion database consists of images from the Plant Image Database of the
Albion College \citep{ALBION} with a fixed size of $1024 \times 685$ pixels.

\item The Calphotos database consists of images from the Regents of the University of 
California \citep{CALPHOTOS} with variable sizes about $700 \times 500$ pixels.
\end{itemize}

The steganographic algorithms used in the experiments are LSB matching \citep{Mielikain:2006}, with embedding bit rates of $0.25$ and $0.10$ bpp, HUGO \citep{Pevny:2010b}, with embedding bit rates of $0.40$ and $0.20$ bpp, and WOW \citep{Holub:2012}, with embedding bit rates of $0.40$ and $0.20$ bpp. 

Since we are addressing targeted steganalysis, the splitting function used for the proposed unsupervised approach is the same steganographic algorithm and the same embedding bit rate, unless otherwise explicitly specified. Note, however, that although we are using the same embedding bit rate, the secret key of the steganographic algorithm (which determines the exact embedded pixels) is not used in our splitting function (since the secret key is assumed unknown by the steganalyzer). 


The steganalysis approach taken in this paper stems from the well-known Rich Models framework \citep{Fridrich:2012}, which proposes a feature extraction step using multiple submodels, and a classification step based on Ensemble Classifiers \citep{Breiman:2001,Kodovsky:2011,Kodovsky:2012}. Ensemble Classifiers scale very well with the number of training samples and dimensionality, and provide similar accuracy compared to the well-known Support Vector Machine (SVM) approach if enough training samples are used \citep{Kodovsky:2011,Kodovsky:2012}. On the other hand, SVMs --particularly the Gaussian Kernel based SVM (G-SVM)-- are more accurate with a relatively small number of samples and dimensions. 

For a fair comparison between the traditional supervised approach and the proposed unsupervised method, we used two different classifications techniques. The supervised approach used the full Spatial domain Rich Model (SRM) feature set and was trained with the BOSS database. In most of the experiments, we used 19,000 images for training (9,500 cover and 9,500 stego), chosen randomly from the full database of 10,000 cover and 10,000 stego images. For a particular experiment (Section \ref{sec:ExperimentIV}), we used a different training set, as explicitly detailed below. 

On the other hand, \textcolor{black}{it must be taken into account that} the proposed method is designed for small sets of images, since we do not have a training set and, therefore, Ensemble Classifiers are not the best choice. For this reason, we used a G-SVM together with a standard feature selection phase (34,671 dimensions are too many for training a G-SVM). More specifically, the feature selection step chooses the best 500 features based on their ANOVA $F$-value. Hence, we provide the best possible settings for both the supervised and the proposed methods, which allows a fair comparison between them.

Whenever SVMs were used, we chose an SVM with a Gaussian kernel. This classifier must be adjusted to provide optimal results. In particular, the values for the parameters $C$ and $\gamma$ must be selected to provide the classifier with the ability to generalize. This process was carried out as described in \citep{Hsu:2003}. For all the experiments with SVMs in this paper, we used cross-validation on the training set applying the following multiplicative grid for $C$ and $\gamma$:
\[
\begin{split} 
C & \in \left\{2^{-5}, 2^{-3}, 2^{-1},2^1,2^3,\dots,2^{15}\right\},\\
\gamma & \in \left\{2^{-15}, 2^{-13}, 2^{-11},\dots,2^{-1},2^1,2^3\right\}.
\end{split}
\]

\subsection{Same Number of Cover and Stego Images and CSM}
\label{sec:ExperimentI}

To compare the results obtained with the supervised method and the proposed approach, we carried out different experiments. As mentioned above, these experiments were performed using Ensemble Classifiers in the case of supervised classification and G-SVM with feature selection for the proposed method. 

In this section, all the experiments with supervised classification were performed training with a database of 19,000 images. The remaining 1,000 images \textcolor{black}{were used as the testing set} when the training and testing databases \textcolor{black}{matched. For the proposed method, all the experiments were carried out} using only the testing set. 

In the first group of experiments, the testing sets consisted of 250 images, 125 of which were cover and the other 125 were stego. The aim of this experiment was to compare how the CSM problem affects the results obtained with supervised classification in contrast with those provided by the proposed method. For this reason, we used the BOSS database as training set for the supervised method and all the databases (including BOSS) for testing. When the testing database is \textcolor{black}{different from BOSS, the CSM problem came along.} 

In Table \ref{tab:comparative_CSM}, we can see the results using LSB Matching with embedding bit rates of 0.25 and 0.10 bpp, HUGO \citep{Pevny:2010b} with embedding bit rates of 0.40 and 0.20 bpp and WOW \citep{Holub:2012} with embedding bit rates of 0.40 and 0.20 bpp. In all the cases, we provide the results using both supervised classification (column ``SUP'') and the proposed method (Artificial Training Sets, column ``ATS'').

It can be observed that the suggested unsupervised approach provided the best results for almost all cases (the best results are boldfaced). Note, also, that the CSM problem is completely avoided with the proposed method, since it does not need a training database, which represents one of the major drawbacks of supervised steganalysis. The worst case occurs with the NRCS database and HUGO with an embedding bit rate of $0.20$ bpp, for which none of the two methods \textcolor{black}{succeeded in classifying} the images correctly (\textcolor{black}{yielding} classification accuracies about $0.5$, i.e., equivalent to random guessing). 

One may think that the results obtained with the suggested method and those of the supervised approach should be almost identical when no CSM problem occurs. However, even when there was no CSM (first row of Table \ref{tab:comparative_CSM}), \textcolor{black}{we obtained} a classification accuracy difference in favor of the suggested approach. We think that the reason for this difference, even when there is no CSM, is the fact that the artificial training step required by our method is performed exactly with the same images of the testing set. Even within the same database, there will be significant differences between different images.  \textcolor{black}{This will lead to relevant differences between the training and the testing sets and degraded classification results when traditional supervised methods are used.} The suggested  ``artificial training'' \textcolor{black}{step} is applied to exactly the same images used for testing, which \textcolor{black}{prevents} this ``degradation''.

Another experiment was conducted using the RANK database. This is the testing database of the Break Our Steganographic System! competition, formed by 1,000 images, which contain both samples from the BOSS database domain and others from \textcolor{black}{a different} domain. Hence, this experiment also exhibits the CSM problem. The results are shown in the first two rows of Table \ref{tab:hugo40_bossrank} for both the supervised approach and the proposed method. In this case the testing set was formed by 471 cover and 471 stego images. We can notice that the accuracy of the proposed method is, again, much higher than that of the supervised counterpart. 

\subsection{Reduced and Unbalanced Number of Stego Samples}
\label{sec:ExperimentII}

In this section, we illustrate a real-world situation of steganalysis, namely, the lack (or a reduced number) of stego samples in the testing set. \textcolor{black}{Generally}, communicating parties that use steganography do not embed information in many images. Therefore, we can probably find enough cover sources from these parties, but \textcolor{black}{relatively fewer} stego images. 


\textcolor{black}{To test this scenario}, we carried out different experiments similar to those presented in the previous section, but using 125 cover images and only 50 stego images in the testing set. After that, other experiments were carried out using 125 cover and only 10 stego images. \textcolor{black}{The tested steganographic systems were LSB matching with $0.25$ and $0.10$ bpp, and HUGO with $0.40$ and $0.20$ bpp.} 

The results for 50 stego images are shown in Table \ref{tab:results_LSBM_50} for LSB matching and in Table \ref{tab:results_HUGO_50} for HUGO. The results using 10 stego images are shown in Table \ref{tab:results_LSBM_10} for LSB matching and in Table \ref{tab:results_HUGO_10} for HUGO. These tables do not only show the classification accuracy (``Acc''), but also the details about true positives (``TP''), true negatives (``TN''), false positives (``FP'') and false negatives (``FN''). In both cases, we can see that  the accuracy decreased compared to that obtained with a greater number of stego samples, \textcolor{black}{but a remarkable level of detection was still achieved}. For example, using the BOSS database and LSB matching with a $0.10$ bpp bit rate, we obtained a classification accuracy of 94\% when the number of stego and cover images was the same (as shown in Table \ref{tab:comparative_CSM}). When we used only 50 stego images, the accuracy decreased to 78\%. \textcolor{black}{Finally,} when we used only 10 stego images, the accuracy was further reduced to 70\%. Although the accuracy decreased, a remarkable level of detection \textcolor{black}{was achieved} even \textcolor{black}{for} only 10 stego images. In particular, we notice that the number of true positives was very high with LSB matching for both cases (10 and 50 stego images), and still quite acceptable for HUGO steganography with 50 stego images. This means that a stego image was very likely to be detected (and further analysis might be carried out to discard false positives). 

The worst situation for the proposed system came out when the number of stego images was reduced to 10. In this case, the results obtained with the supervised approach were better than those of the proposed method when no CSM occurs (the first row of Tables  \ref{tab:results_LSBM_10} and \ref{tab:results_HUGO_10}). In case of CSM (the other rows), the proposed system provided better results than the supervised method for some testing databases and worse results for others. This illustrates that  the performance of the proposed system decreases when the ratio between the number of stego images and the total number of testing images is very small. This situation is \textcolor{black}{further analysed by means of} other experiments in this section.

We also performed an experiment using the RANK database with two different ratios between cover and stego images: 70\%/30\% (471/202) and 90\%/10\% (471/52). The results, shown in Table \ref{tab:hugo40_bossrank}, illustrate that the proposed method provided excellent detection accuracy (greater than or equal to 80\%) for these cases. This shows that the accuracy is still high for unbalanced testing sets. The next experiment was performed in order to determine the threshold to obtain convenient detection results \textcolor{black}{(classification accuracy) with the proposed approach as the ratio between stego and cover images is considered.}

Since the classification accuracy seems to decrease when we have an unbalanced testing set, \textcolor{black}{considering the relative number of cover and stego images,} we carried out the next family of experiments to investigate how the proposed method behaved with different ratios of stego images. In Table \ref{tab:hugo40_step5}, the detection accuracy obtained with an increasing ratio of stego images \textcolor{black}{is shown for a testing set taken from the BOSS database embedded using HUGO with $0.40$ bpp.} We began with a testing database of 125 cover and 0 stego images (125/0) and, in each step, we removed five cover and added five stego images, until a testing set of 0 cover and 125 stego images (0/125) was \textcolor{black}{formed} for the last experiment. We can see that the proposed ATS method provided convenient results (accuracy over 70\%) when the ratio of cover and stego images was roughly between 10\% and 95\%. The accuracy decreased when either the number of stego or cover images was very small (less than 10\% of stego images or less than 5\% of cover images). A particularly difficult situation occurred when the testing set was formed exclusively by cover images. Even in that case, the \textcolor{black}{proposed} approach tried to separate the testing set into cover and stego images, leading to poor classification results (only 26\% of the testing images were correctly classified as cover images). This situation shall be prevented, since testing sets formed only by cover images can be very usual in real-world scenarios. This problem shall be addressed in the future research to avoid  \textcolor{black}{such a high level} of false positives. 


Another interesting real-world scenario comes along when the total number of testing images is very small. This situation may occur, for example, when \textcolor{black}{only} a few images are found in a USB stick and they have to be evaluated by a steganalyst. \textcolor{black}{In order to} test this scenario, we carried out a set of experiments using very small sets \textcolor{black}{(with only 20, 15, 10 and 8 images)} with variable cover/stego ratios.  In Table \ref{tab:results_few1}, the results obtained with 1) ten cover and ten stego images, and 2) five cover and five stego images, \textcolor{black}{are shown}. Similarly, in Table \ref{tab:results_few2}, we can see  the results obtained with 1) five cover and ten stego images, and 2) ten cover and five stego images. Finally,  Table \ref{tab:results_few3}, \textcolor{black}{shows} the results obtained with 1) five cover and three stego images, and 2) three cover and five stego images. The algorithms and embedding bit rates used in the experiments were LSBM with 0.25 bpp and HUGO with 0.40 bpp. The results shown in the tables are the average values obtained after repeating each experiment ten times with a different selection of images. Because of this, the true and false positive and negative values are not integers. \textcolor{black}{It is worth pointing out that the results are remarkable even when the number of stego/cover images were 3/5 and 5/3,} with detection accuracies close to or higher than 70\%, except for the NRCS database. Taking into account the reduced size of these testing sets, we can conclude that the reliability of the proposed system is quite \textcolor{black}{noteworthy}.

\subsection{Testing Set from Mixed Databases}
\label{sec:ExperimentIII}

This section presents an even more challenging situation for the proposed steganalysis approach. In the next experiments, the images (both cover and stego) came from different databases with an uneven proportion. The mixed data set was formed by 280 images as follows: 70 images from BOSS, 60 images from INTE, 50 images from CALP, 40 images from ESO, 30 images form ALBN, 20 images from NOAA and ten images from NRCS. Then, we selected 140 of these images as cover and the other 140 images as stego, randomly. After that, we embedded a message in the 140 images selected as stego using LSB matching with an embedding bit rate of $0.10$ bpp, and removed the corresponding original (cover) images from the testing set. We can see the results after applying the proposed steganalysis method in the first row of Table \ref{tab:results_LSBM_10_MIX}, which shows that the classification accuracy is still quite large (79\%), despite the hostile scenario.

After that, we applied the same procedure as above but using 50 stego images only (and a total of 190 images between cover and stego). In this case, the results obtained after applying the proposed steganalysis method are shown in the second row of Table \ref{tab:results_LSBM_10_MIX}. Since the images in the different databases have different sizes, we repeated the same experiment but with all the images clipped to $512\times 512$ pixels, as shown in the last two rows of the table.

We can see that the only problem \textcolor{black}{with the proposed approach was} a relatively high number of false positives. However, even in this challenging situation, the accuracy of the suggested method was still above 70\%. Again, the most remarkable property of these results is the ability of the method to provide a very large number of true positives. Note that using images with the same size the accuracy increases about 10 percentage points for the proposed method and 20 percentage points for the supervised approach. In any case, the best classification results are always those provided by the proposed method.

On the other hand, the supervised approach, trained with the BOSS database consisting of 19,000 images, does not yield good classification results, probably due to the CSM problem. Since the other databases do not have that many images, including them in the training set would not be a convenient strategy, since the training set would be unbalanced (it would contain many more images from BOSS than from the other databases). 
Nevertheless, there are some techniques in the state of the art that allow to deal with the CSM problem with supervised classification, such as the methods of \citep{Pasquet:2014}, in which a clustering approach is used, and \citep{Lubenko:2012}, in which millions of images are used for training.

\subsection{Testing Set with Mixed Embedding Bit Rates}
\label{sec:ExperimentIV}

As detailed in Sections \ref{sec:ATS_method} and \ref{sec:TheoreticalAnalysis}, the proposed approach requires using the targeted steganographic system as a splitting function \textcolor{black}{with} approximately the same embedding bit rate. The experiment presented in this section uses three of the selected databases and the stego images were generated with the same steganographic method (LSB matching) but with five different embedding bit rates: $0.25$, $0.20$, $0.15$, $0.10$ and $0.05$ bpp. We took 250 images for the testing set: 125 cover images (without any change) and 125 stego images obtained using the embedding method and one of the five possible embedding bit rates. Once embedded, the original sources were removed from the testing set. The stego image set was thus formed by 25 images for each embedding bit rate.

The \textcolor{black}{main} difficulty with addressing this steganalysis problem is \textcolor{black}{the fact that} we do not know the embedding bit rate. Even worse, the embedding bit rate is different for different images in the stego set. However, the suggested approach requires a splitting function that is approximately the same that the targeted steganographic method. Thus, we need to choose an appropriate splitting function. 

It can be easily understood (from the analysis presented in Section \ref{sec:TheoreticalAnalysis}) that using LSB matching with a low embedding bit rate will not conveniently separate the testing set. A stego image with a $0.10$ bpp embedding bit rate would have 10\% of pixels selected for embedding and, on average, half of them (a 5\% of the total) would be modified compared to the corresponding cover image (the other half will already have the correct value in the LSB). If this stego image is embedded again, with an embedding bit rate of $0.10$ bpp, the number of modified pixels with respect to the corresponding cover image would be, approximately, $5\% \text{ (already modified)} + 95\% \cdot 5\% \text{ (modified by the second embedding)}= 9.75\%$. This number is not exact, since some already modified pixels could be reverted to the original value in the second embedding, but it suffices to illustrate the problem. Similarly, the number of pixels modified for a cover image embedded with $0.20$ bpp is 10\% on average. The ``double stego'' image (embedded with $0.10$ bpp twice) will have differences of up to $\pm 2$ for a few pixels (about a $0.25\%$ of them), whereas a $9.5\%$ of them will have differences of $\pm 1$ compared to the cover image. Hence, it would be very difficult to separate ``double stego'' images with $0.10$ bpp from stego images with $0.20$ bpp. This means that the condition $X' \cap s[X']=\emptyset$ would not be satisfied. 

On the other hand, a too large embedding rate, e.g. $0.50$ bpp, would make it difficult to separate the $0.05$ bpp stego images from the cover images: a $0.05$ bpp stego image embedded again with a $0.50$ bpp bit rate would modify approximately $2.5\%+97.5\%\cdot 25\%=26.875\%$ of pixels on average, compared to the corresponding cover image, most of them with a difference of $\pm 1$, whereas a cover image embedded with a $0.50$ \textcolor{black}{bpp} bit rate modifies 25\% of the pixels, on average, with a difference of $\pm 1$.

Hence, the selection of the appropriate embedding bit rate for the splitting function in this situation is critical. To overcome this difficulty, for this particular experiment, we used LSB matching with $0.25$ bpp as a splitting function to check if the separation property could still be achieved. 

\color{black}

The results, shown in Table \ref{tab:results_LSBM_MBR}, illustrate that this approach leads to excellent  classification results, greater than 85\% in accuracy, for three different databases. For the supervised method, we used a training set formed with 9,500 cover images, and $5\cdot \text{9,500}=\text{47,500}$ stego images, since we embedded the chosen images once for each tested embedding bit rate (i.e. $0.25$, $0.20$, $0.15$, $0.10$ and $0.05$ bpp). The training was carried out with the BOSS database and, hence, the results of the last two rows of the table also exhibit the CSM problem. It must be pointed out that, although the training set contained stego images embedded with all the tested bit rates, the supervised approach could not classify the testing set even when there was no CSM problem (first row).

This experiment shows that the proposed method can still be applied even if we do not know the exact embedding bit rate of the targeted steganographic system. However, a method for selecting the optimal embedding bit rate for the splitting function is required and must be addressed in the future research.

\color{black}

\subsection{Testing Set with Unknown Message Length}
\label{sec:ExperimentV}

The usual scenario in steganography assumes, by the Kerckhoffs' principle, that the steganalyst knows all details about the steganographic channel except the secret keys. However, in a real-world scenario, some details, such as the embedding \textcolor{black}{bit} rate, are rarely known. This problem is often referred to as detecting messages of unknown length \citep{Pevny:2011}.

In the state of the art, there are quantitative steganalyzers that try to discover the embedding \textcolor{black}{bit} rate. In \citep{Pevny:2011} and \citep{Kodovsky:2013}, the authors use a regression algorithm for obtaining a mapping $F: \mathcal{F} \mapsto \mathcal{P} \subseteq \mathbb{R}$, where $\mathcal{F} \equiv \mathbb{R}^n$ is a feature space representation of images and $\mathcal{P}$ is a compact subset representing the space of embedding \textcolor{black}{bit} rates. Unfortunately, there is not a direct way of applying this framework in the proposed method. The suggested system is unsupervised and there is no training set with which we can train the regression algorithm. Despite that, in this section, we present some experiments about how to deal with a testing set of images for which the embedding bit rate is not fully known, which is a challenging situation for the proposed system. \textcolor{black}{Below}, we propose a methodology based on the measure of distance between the centroids of the different classes. 

As a result of applying the proposed method (see Section \ref{sec:ATS_method} and Fig.\ref{fig:ATS_method}), we obtain a classification of $B$ into stego and ``double stego'' images and, by the existing bijection between the elements of $B$ and $A$, a classification of $A$ into cover and stego images. Therefore, if the classification is successful, we expect that the elements of the stego part of $A$ be similar to those of the stego part of $B$ (see Fig.\ref{fig:MEd}). Likewise, we expect that the elements of the cover part of $A$ be dissimilar to those of the stego part of $A$ (see Fig.\ref{fig:MEa}). To exploit this idea, we can calculate the centroid of the cover part of $A$, namely $C_{A_\mathrm{cover}}$, the centroid of the stego part of $A$, namely $C_{A_\mathrm{stego}}$, and the centroid of the stego part of $B$, namely $C_{B_\mathrm{stego}}$, \textcolor{black}{and compute the distances between them: $d(C_{A_\mathrm{stego}}, C_{B_\mathrm{stego}})$ and $d(C_{A_\mathrm{cover}}, C_{A_\mathrm{stego}})$.}  \textcolor{black}{This computation can be carried out using a standard distance measure $d$,} such as the  Euclidean distance:
$$d(x,y)=\sqrt{\sum_{i=0}^{n}(x_i-y_i)^2}.$$ 

We expect short distances for $d(C_{A_\mathrm{stego}}, C_{B_\mathrm{stego}})$, because they are similar, and large distances for $d(C_{A_\mathrm{cover}}, C_{A_\mathrm{stego}})$, since they are dissimilar. Then, we can compute a score value $S$ as follows:

$$S=\frac{d(C_{A_\mathrm{stego}}, C_{B_\mathrm{stego}})}{d(C_{A_\mathrm{cover}}, C_{A_\mathrm{stego}})}.$$ 

This value will be small when the classification is correct, that is, when we have managed to guess the correct algorithm and embedding bit rate. Therefore, the steganalyst can select a list of tentative bit rates and try them \textcolor{black}{sequentially.} The bit rate that provides the lowest score will most \textcolor{black}{possibly} be the right one (or close to it).

Although there are several open questions with respect to this methodology, which shall be addressed in the future research, some results are shown in Table \ref{tab:boss_unknown}. These experiments were performed using 200 images of the BOSS database, 100 of which were cover and 100 stego. Some tests were also carried out with only cover images. In that case, we only used 100 images in the testing set.  The distances were calculated using 50 features, selected from the best $F$-values of an ANOVA test. The details of the different experiments are discussed below.

We carried out experiments with an unknown message length for three different \textcolor{black}{steganographic} algorithms: LSB matching, HUGO and WOW, with an embedding \textcolor{black}{bit rate of} $0.40$ bpp in all cases. The steganalyst needed to find out what the real embedding bit rate was and, for this purpose, he/she prepared a list of tentative bit rates, namely $0.10$, $0.20$, $0.30$, $0.40$, $0.50$ and $0.60$ bpp. The procedure is straightforward: the steganalyst only had to apply the proposed method one time for each tentative bit rate and keep the results with lowest score. As shown Table \ref{tab:boss_unknown}, in all three cases, the lowest score indicated the \textcolor{black}{correct} bitrate. 

\color{black}

\subsection{Real-time Construction of the Testing Set}
\label{sec:ExperimentVI}

Another relevant scenario to consider for the proposed scheme is how to proceed when we do not have access to the whole set of images (formed with stego and cover samples) at the same time. This can be the typical situation that arises when we are eavesdropping the communications between two (or more) parties. If these parties exchange a few images (or even one image) from time to time, it is not possible to run the proposed classification method with the complete set of images, as shown in the previous sections. 

In this case, the testing set $A$ will be built dynamically. To deal with this realistic scenario, we propose the following procedure:

\begin{enumerate}
\item Collect the images of the testing set $A$ one by one until $|A|\geq n_{\min}$, where $n_{\min}$ is a minimum number of images required to apply the method. Hence, the set $A=\{I_1,I_2,\dots,I_{n_{\min}}\}$ is built, where $I_j$ stands for the image obtained at the $j$-th iteration.
\item \label{classify} Classify the set $A$ into cover ($V$) and stego ($W$) images applying the proposed method. Output the label of each image (``cover'' or ``stego'').
\item When a new image $I_k$ is obtained, add the new image into the set $A:=A \cup \{I_k\}$, and repeat Step \ref{classify}.  
\end{enumerate}
As shown in Section \ref{sec:ExperimentII}, $n_{\min}$ can be very small and still provide remarkable detection accuracy. In the experiments presented below, the value $n_{\min}=10$ has been used.

Although the strategy of repeating the classification with the whole set may appear simplistic, in fact, it is the best option from the accuracy point of view. As shown below by means of several experiments, the classification accuracy increases with the number of classified images. Hence, when a new image is obtained, it is better to classify the whole set of images again and output the result of the last classification. Needless to say, this strategy also requires more computational effort. Accuracy is thus obtained in exchange for computational cost. We consider this strategy realistic, since the classification of even hundreds of images can be carried out in just a few minutes with standard hardware. If reduced computation time is a strong requirement in a real-time implementation of the scheme, accuracy can be sacrificed by selecting a subset of the collected images for each classification.

In addition, this re-classification of the whole testing set each time a new image is obtained produces a side effect: for each image, we do not only have the last classification result (as ``cover'' or ``stego'') but also the results obtained in all the previous classifications (iterations). In fact, assuming that we obtain the images one by one, if $n=|A|$ is the current number of images (i.e., $I_n$ is the last image that has been included into the set $A$), we have the following number of classification results (labels), $n_l$, per each image: 
\[
n_l(I_k) = \left\{\begin{array}{ll}
\max(0,n-n_{\min}+1), & \text{if }k\leq n_{\min},\\
\max(0,n-k+1), & \text{otherwise}.
\end{array}\right.
\]
Now, we can compute $m_l(I_k)$ as the number of times that each image has been classified \textbf{with the same label} as in the last classification experiment. In the best case for classification ``confidence'', we will have $m_l(I_k)=n_l(I_k)$, that is, the image $I_k$ has been classified with the same label in all the iterations. In the worst case, we will have $m_l(I_k)=1$, meaning that the previous $n_l(I_k)-1$ classification experiments yielded the opposite label compared to the current experiment. With these considerations, we can define a simple ``confidence level'' for the classification result of each image:

\begin{equation}
c(I_k)=\frac{m_l(I_k)}{n_l(I_k)}, \text{ if } n_l(I_k)>0.
\label{eq:confidence}
\end{equation}

The value of the confidence level provided in Equation \ref{eq:confidence} satisfies $c(I_k) \in (0,1]$. The closer $c(I_k)$ to 1, the more confident we can be about the classification of $I_k$ and, conversely, the closer $c(I_k)$ to 0, the less confident we can be about the label assigned to that image. One of the advantages of such a measurement is that it can be computed without knowledge of the true category of each image. The confidence level of the classification of an image $I_k$ always begins with $c(I_k)=1$ the first time it is classified, and it will typically have a lower value as it is re-classified each time a new image is included into the set $A$. The values provided by this indicator are also analysed in the experiments below.

\begin{figure}[!ht]
  \subfloat[124 cover and 31 stego images]{\label{fig:RTa}
  \ifpdf \includegraphics[width=0.45\textwidth]{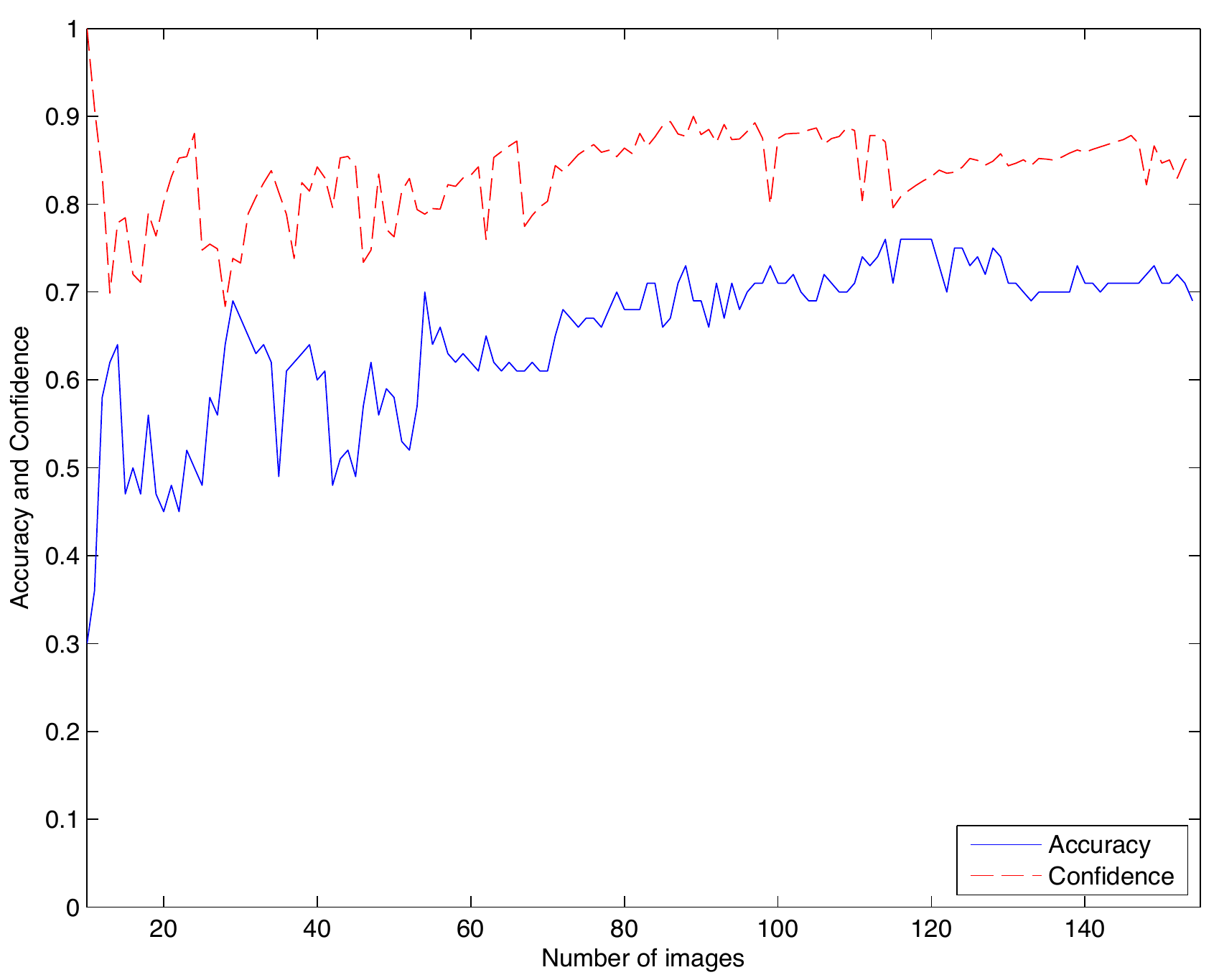}
  \else \includegraphics[width=0.45\textwidth]{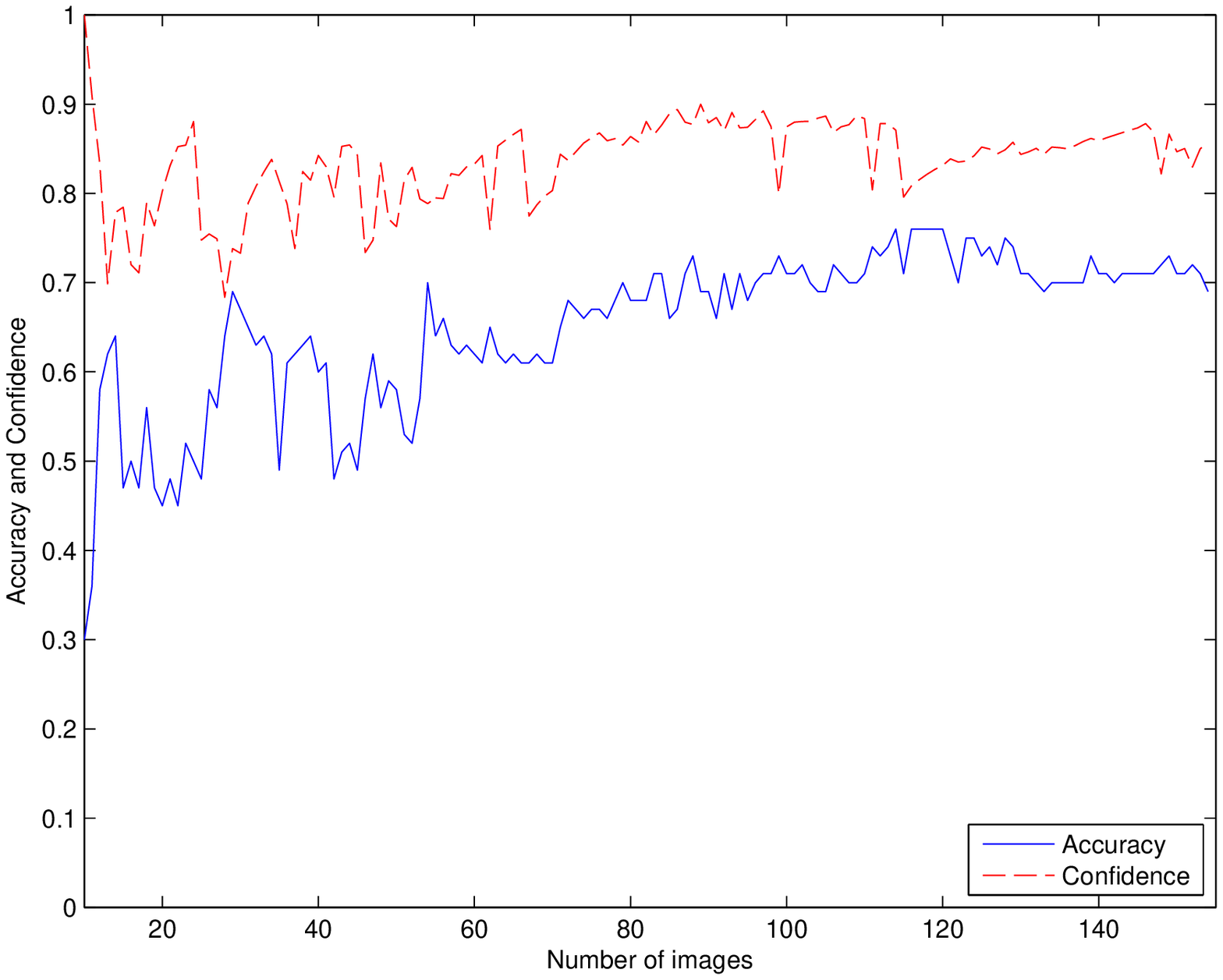}
  \fi} \qquad
  \subfloat[105 cover and 45 stego images]{\label{fig:RTb}
  \ifpdf \includegraphics[width=0.45\textwidth]{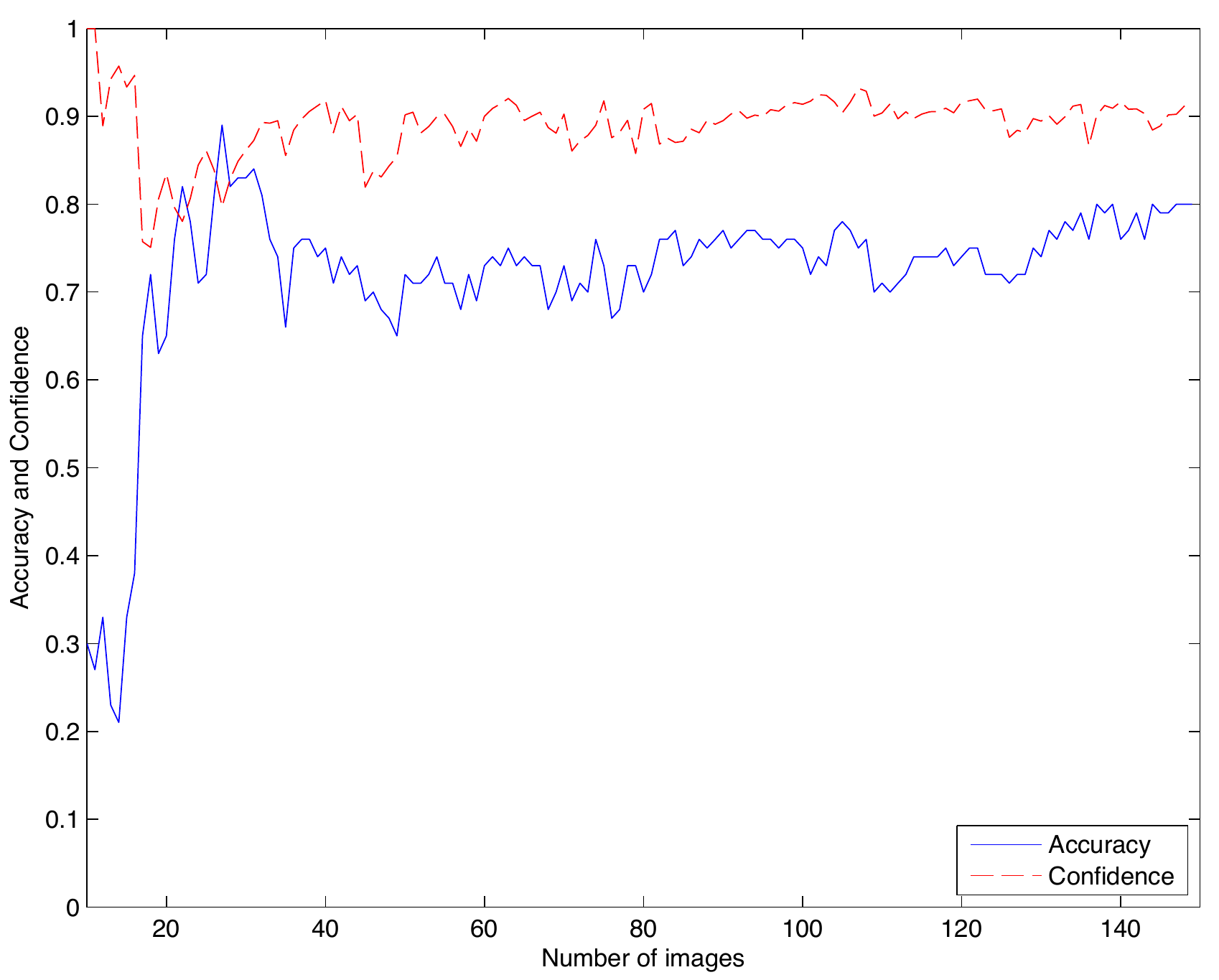}
  \else \includegraphics[width=0.45\textwidth]{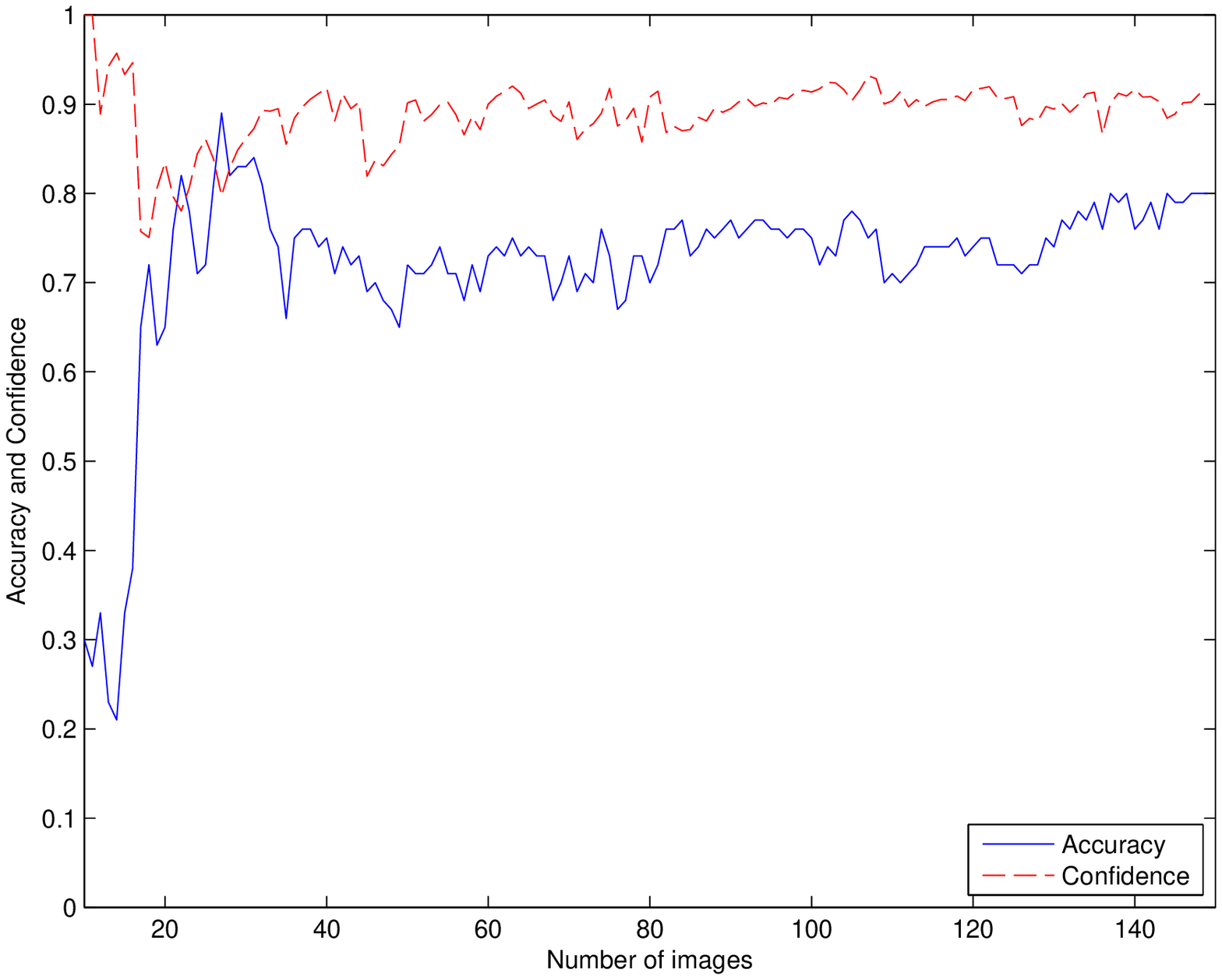}
  \fi} \\
  \subfloat[125 cover and 125 stego images]{\label{fig:RTc}
  \ifpdf \includegraphics[width=0.45\textwidth]{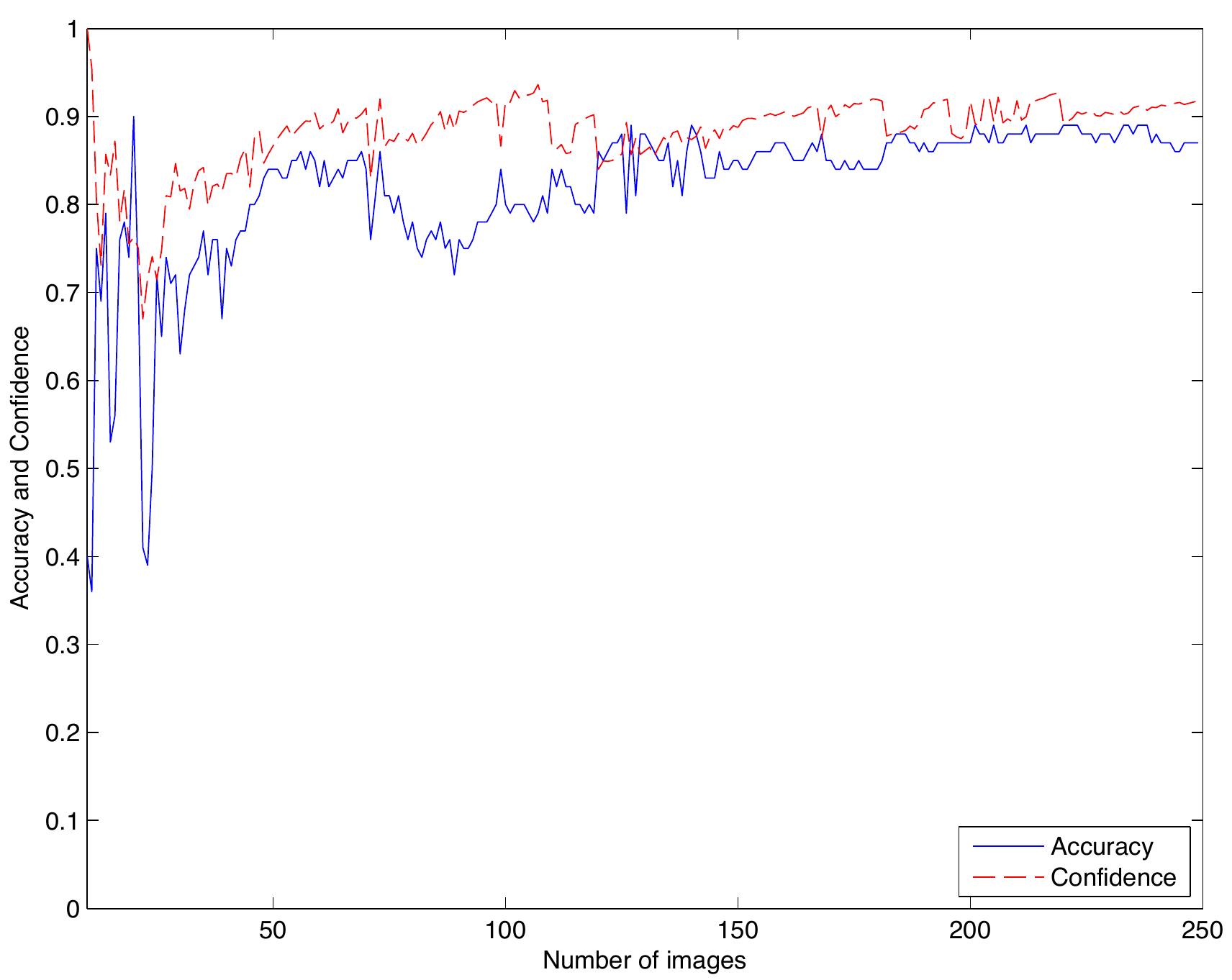}
  \else \includegraphics[width=0.45\textwidth]{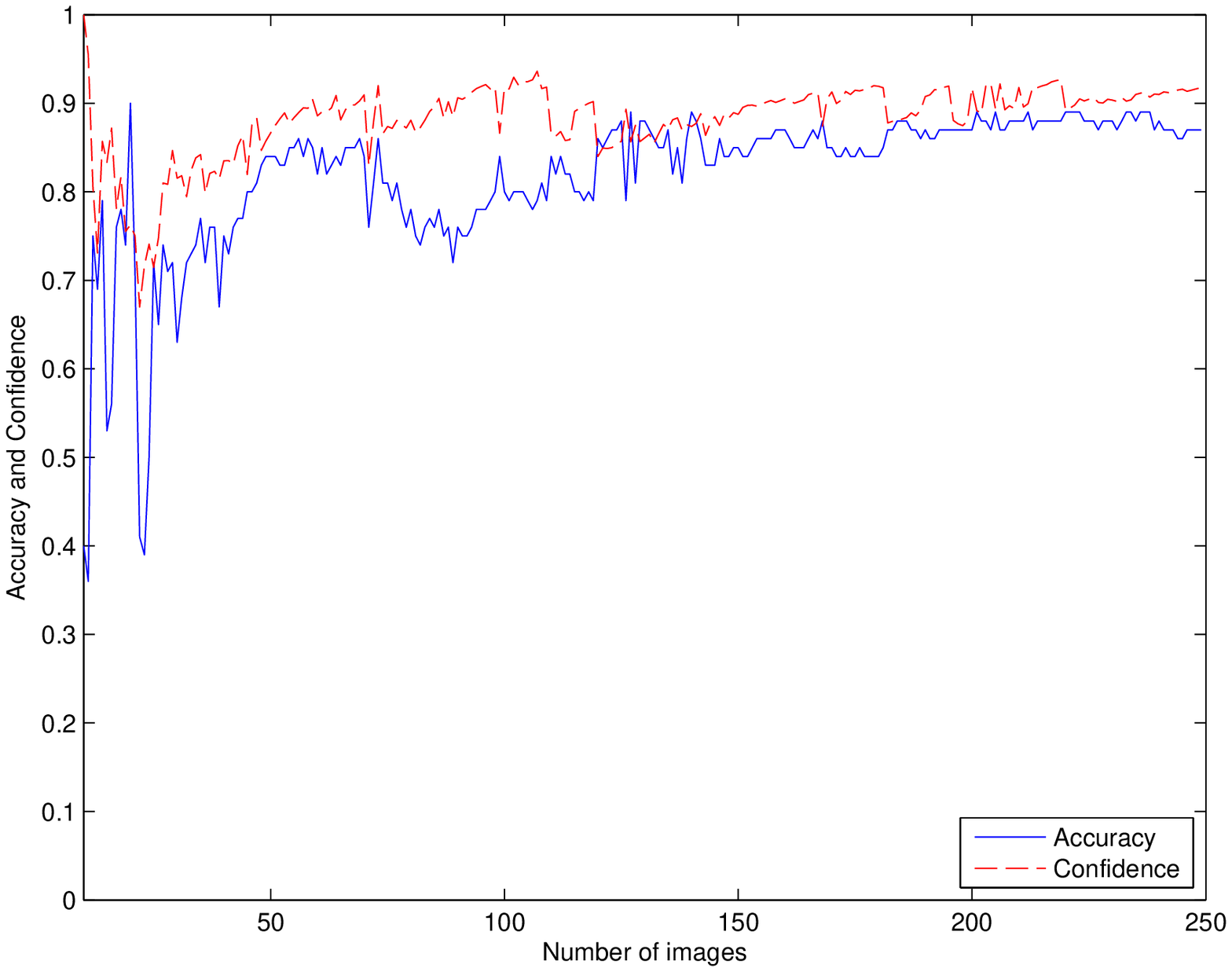}
  \fi} \qquad 
  \subfloat[105 cover and 45 stego images (average of 100 simulations)]{\label{fig:RTd} 
  \ifpdf \includegraphics[width=0.45\textwidth]{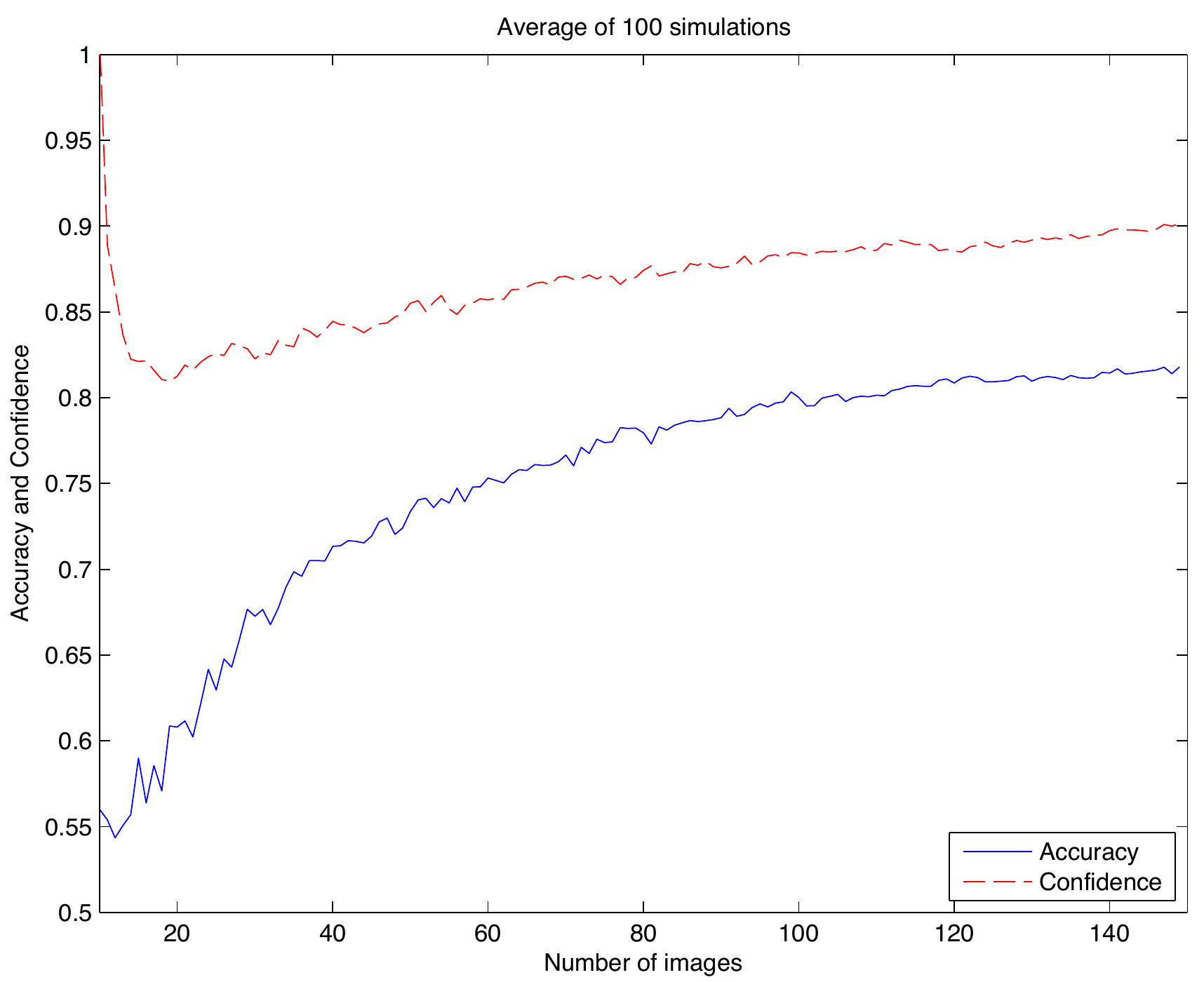} 
  \else \includegraphics[width=0.45\textwidth]{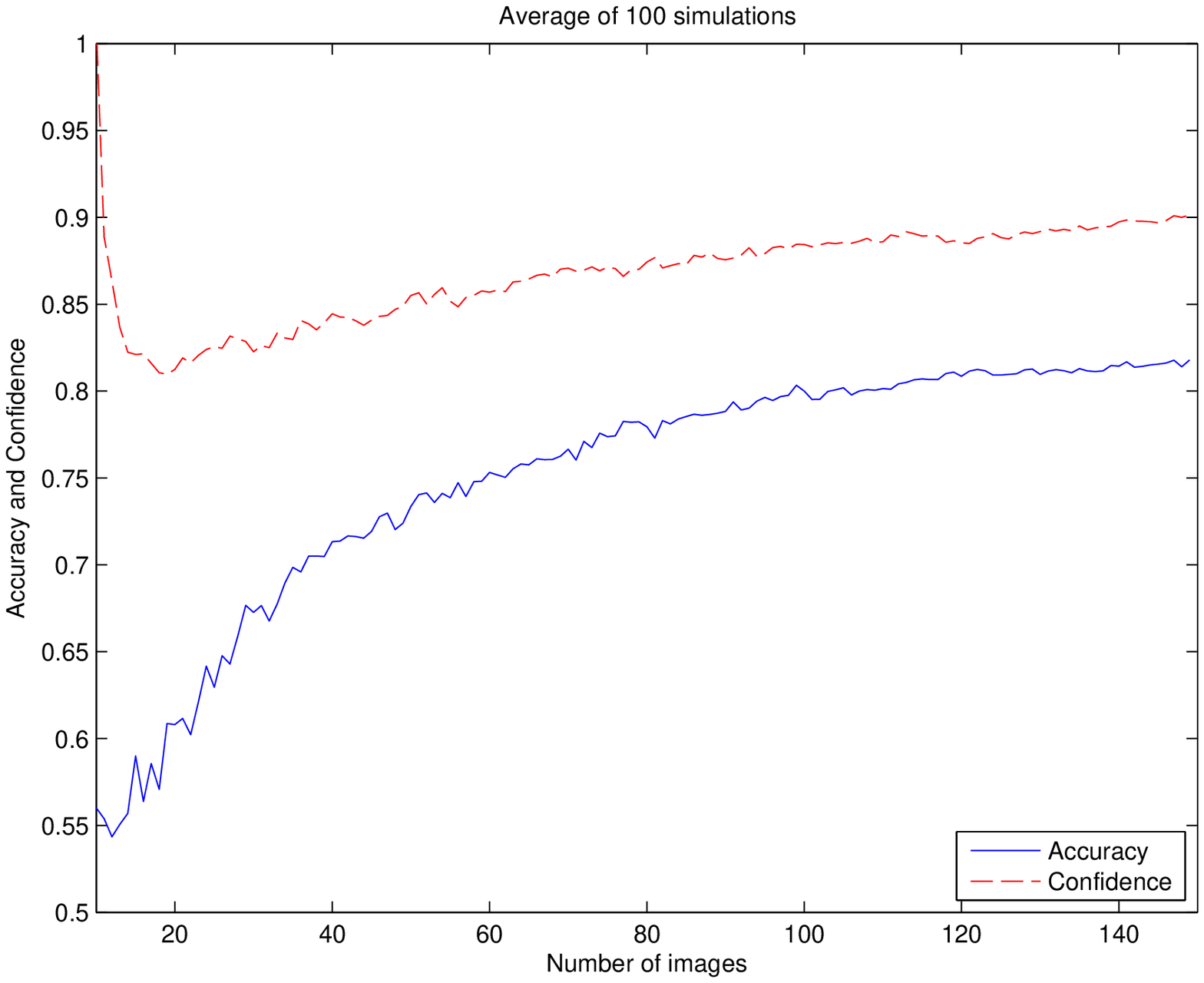}
  \fi}  
  \caption{Accuracy (solid line) and confidence level (dashed line) for testing sets constructed in real time for the BOSS database using HUGO with $0.40$ bpp}
  \label{fig:realtime}
\end{figure}

The first experiment to test this scenario was simulated taking 155 images, 124 of which were cover and the remaining 31 have been embedded using HUGO steganography with $0.40$ bpp (i.e. 20\% of the images were stego). At each iteration, a random image was selected from the set and added into $A$. As remarked above, we set $n_{\min}=10$. Hence, the first classification experiment was carried out when the 10-th image was selected. The results, in terms of average classification accuracy and average confidence level are shown in Fig.\ref{fig:RTa}, using a solid line for accuracy and a dashed one for confidence. At each iteration, the average accuracy and the average confidence were computed for the set of available images. As accuracy is concerned, the classification results with a few images (approximately less than 50) were low, but when the number of images reaches this threshold, the accuracy increased, reaching a 70\% ($0.7$) in the last iteration. On the other hand, the confidence, started from a high value, showed a significant decrease in the first few iterations and at some iteration (again about 50 images) started a quite regular increase. We can also observe a strong correlation between both curves, specially when the number of collected images was significant enough. This situation is really remarkable, since in a real experiment we would not know the real accuracy, but we would be able to compute the confidence measurement, since it only depends on the past classification results. When the confidence reaches a ``steady state'' behavior, the same occurs with the accuracy. Hence, this provides with a strong indicator of the quality of the classification results.

The same procedure was applied for other testing sets of images. Fig.\ref{fig:RTb} shows the results obtained with 150 images, 105 of which were cover and the remaining 45 stego (i.e. 30\% of stego images). The behavior of both variables was similar to that of the previous experiment, but the threshold in the number of images required for a high accuracy and confidence was lower (about 30 images). In this case, the final accuracy was also higher, about 80\% ($0.8$).

Fig. \ref{fig:RTc} provides the results obtained for 250 images, 125 of which were stego (i.e. 50\% of stego images). The situation was similar to that of the previous two experiments. From a certain iteration (again about 30 images) the accuracy and confidence curves showed a more regular behavior, mostly increasing. In this case, the final accuracy was higher, close to 90\% ($0.9$), due to the larger number of images in the final testing set (250 images, compared to 155 for the first experiment and 150 for the second one).

The correlation between accuracy and confidence could also be observed for the iterations in which there was a sudden variation in both variables. When accuracy shows a sudden increase or decrease, a similar situation occurs with confidence, but the increase or decrease was sometimes reversed compared to accuracy. There is a simple explanation for this situation. A large variation in confidence means that many images changed their label at that iteration. This change of label could be right, leading to an increased accuracy, or wrong, leading to a decrease in accuracy. Only occasionally, these variations were compensated as right and wrong classification results (leading to small variations in accuracy). It can also be observed that the higher the number of images, the more likely the variations in accuracy and confidence had the same sign. This explains the strong positive correlation between both curves when the number of images reaches a minimum threshold (about 30-50).

The results shown in Figs.\ref{fig:RTa}-\ref{fig:RTc} provide only one simulation in each case. Such a simulated experiment tries to reproduce the situation that we would find in a real-life experiment, i.e. the images would be obtained one by one, or in a small number at each iteration. However, the results obtained in this way can be biased due to the generation of the pseudo-random numbers used to select the image order. In order to avoid this bias, we carried out 100 simulations with the same set of 105 cover and 45 stego images, but using a different seed for the pseudo-random selection (ordering) of the images at each simulation. Fig.\ref{fig:RTd} shows the averaged results of these 100 simulations for both accuracy and confidence (please note the scale change in the vertical axis with respect to the other three cases). In this figure, we can notice that accuracy became a much ``softer'' curve, starting from low values (about 55\%) and reaching high accuracy (over 80\%) for the complete set. The increase in accuracy is almost monotonic, and the oscillations were almost suppressed by averaging the results of 100 experiments. In addition, we can also observe the same ``softer'' shape of the confidence measurement. The correlation between accuracy and confidence when the number of images was larger than 30 was almost perfect. Both curves increased approximately linearly, with roughly the same slope. 

These experiments suggest that using all the available images increases the accuracy and justifies the method proposed for a real-time implementation of our scheme. No accuracy gain can be expected, a priori, from discarding some images from the testing set.

It must be taken into account that if the ratio of stego images varies too much in real-time, the value of accuracy may not increase as regularly as shown in Fig.\ref{fig:realtime}. As discussed in Section \ref{sec:ExperimentII}, if the percentage of stego images is below 10\% or above 95\%, the detection accuracy results may decrease significantly. This situation may occur if, from some given moment, (nearly) only cover or (nearly) only stego images were obtained, leading to a very low or a very high percentage of stego images. This kind of scenario must be prevented to maintain the accuracy results in acceptable values.

\color{black}

\section{Conclusion}
\label{sec:Conclusion}

In this paper, a novel unsupervised steganalysis method is presented. We show how unsupervised steganalysis can be addressed by using an artificial training set and supervised classification if we know the algorithm and the embedding bit rate used for steganography. Hence, the suggested method is applicable in targeted steganalysis. Using the proposed approach, we can also bypass the CSM problem and outperform the state-of-the-art methods. Removing the necessity of a training data set in the machine learning problem is the major contribution of this paper.

The proposed approach has been tested using three steganographic methods: LSB matching, HUGO and WOW. It is shown that we can achieve better classification accuracy than that obtained using traditional supervised steganalysis (Rich Models, Ensemble Classifiers and SVM), while avoiding the CSM problem that makes the performance of supervised steganalysis decrease significantly.

Furthermore, through the different experiments presented in the paper, we show that the proposed method can address complex real-world situations, in which we do not have a clear training database (e.g., when the images come from different databases), the number of stego images is reduced \textcolor{black}{or the images are obtained one by one, in real time.} We show that the suggested method provides remarkable performance even if the images are selected unevenly from different databases or if the embedding bit rate is unknown and variable for different stego samples.

As future work, it would be worth researching how the suggested method can be applied in situations for which we do not know the image database, the steganographic algorithm or the \textcolor{black}{embedding} bit rate used, as it would occur in real-world steganalysis. The experiments with images taken unevenly from different databases show some decrease in the accuracy results of the method. Besides, if the message length is unknown, the approaches presented in the paper shall be analyzed more deeply, though the preliminary results are promising. Similarly, the case when the testing set is formed only by cover images needs be specifically addressed to avoid a large number of false positives. In addition, it would also be interesting to investigate how to proceed when we cannot estimate the splitting function, e.g. when we deal with a novel and unknown steganographic scheme. This would mean porting the proposed approach to the problem of universal steganalysis.

Finally, we propose to investigate the application of this approach in other fields, beyond steganalysis, with similar properties. The idea behind the proposed method could be exploited whenever a splitting function can be defined in a machine learning classification problem.

\section*{Acknowledgment}                             
This work was partly funded by the Spanish Government through grants TIN2011-27076-C03-02 ``CO-PRIVACY'' and TIN2014-57364-C2-2-R ``SMARTGLACIS''.

\bibliographystyle{elsarticle-harv}

\bibliography{main}

\newpage
\clearpage

\begin{table}[ht]
\begin{center}
\begin{tabu}{l|cc|cc|cc|cc|cc|cc}
\hhline{=============} 
\multicolumn{13}{c}{Comparative analysis between supervised classification and ATS} \\ 
\hhline{-------------} 
\multicolumn{1}{c}{} & 
\multicolumn{4}{|c}{LSB Matching} &
\multicolumn{4}{|c}{HUGO} &
\multicolumn{4}{|c}{WOW} \\
\hhline{~------------} 
\multicolumn{1}{c}{} & 
\multicolumn{2}{|c}{$0.25$ bpp} & 
\multicolumn{2}{|c}{$0.10$ bpp} &
\multicolumn{2}{|c}{$0.40$ bpp} & 
\multicolumn{2}{|c}{$0.20$ bpp} &
\multicolumn{2}{|c}{$0.40$ bpp} & 
\multicolumn{2}{|c}{$0.20$ bpp} \\

\multicolumn{1}{c}{Test DB} & 
\multicolumn{1}{|c}{SUP} & 
\multicolumn{1}{c}{ATS} &
\multicolumn{1}{|c}{SUP} & 
\multicolumn{1}{c}{ATS} &
\multicolumn{1}{|c}{SUP} & 
\multicolumn{1}{c}{ATS} &
\multicolumn{1}{|c}{SUP} & 
\multicolumn{1}{c}{ATS} &
\multicolumn{1}{|c}{SUP} & 
\multicolumn{1}{c}{ATS} &
\multicolumn{1}{|c}{SUP} & 
\multicolumn{1}{c}{ATS} \\

\hline \hline

BOSS & 
0.96 & \textbf{0.98} & 0.90 & \textbf{0.94} & 
0.78 & \textbf{0.87} & 0.67 & \textbf{0.79} & 
0.82 & \textbf{0.94} & 0.62 & \textbf{0.82} \\

ESO & 
0.37 & \textbf{0.94} & 0.50 & \textbf{0.98} & 
0.50 & \textbf{0.83} & 0.45 & \textbf{0.86} & 
0.53 & \textbf{0.98} & 0.50 & \textbf{0.94} \\

CALP & 
0.51 & \textbf{0.96} & 0.57 & \textbf{0.94} & 
0.48 & \textbf{0.90} & 0.48 & \textbf{0.84} & 
0.52 & \textbf{0.95} & 0.51 & \textbf{0.85} \\

INTE & 
0.47 & \textbf{0.98} & 0.50 & \textbf{0.98} & 
0.51 & \textbf{0.95} & 0.51 & \textbf{0.95} & 
0.50 & \textbf{1.00} & 0.50 & \textbf{0.96} \\

NRCS & 
0.49 & \textbf{0.68} & 0.55 & \textbf{0.61} & 
0.49 & \textbf{0.63} & \textbf{0.50} & 0.46 & 
0.50 & \textbf{0.86} & 0.50 & \textbf{0.68} \\

ALBN & 
0.57 & \textbf{0.99} & 0.65 & \textbf{0.98} & 
0.50 & \textbf{0.97} & 0.50 & \textbf{0.91} & 
0.54 & \textbf{0.95} & 0.49 & \textbf{0.92} \\

NOAA & 
0.35 & \textbf{1.00} & 0.50 & \textbf{0.98} & 
0.50 & \textbf{1.00} & 0.44 & \textbf{0.96} & 
0.50 & \textbf{0.98} & 0.50 & \textbf{1.00} \\

\hline \hline 
\end{tabu}
\end{center}
\caption{Accuracy of supervised classification \textcolor{black}{(``SUP'')} compared with the proposed method (``ATS'') \textcolor{black}{for different image databases (``DB'')} embedded with different algorithms and bit rates}
\label{tab:comparative_CSM}
\end{table}

\begin{table}[ht]
\begin{center}
\begin{tabu}{ccccccc}
\hhline{=======} 
\multicolumn{7}{c}{BOSS RANK Experiments} \\ 
\hhline{-------} 
\multicolumn{7}{c}{\textbf{HUGO $\mathbf{0.40}$ bpp}} \\
\hhline{-------} 
\textbf{\scriptsize Method} & 
\textbf{\scriptsize Cover/Stego} & 
\textbf{\scriptsize Acc} & 
\textbf{\scriptsize TP} & 
\textbf{\scriptsize TN} & 
\textbf{\scriptsize FP} & 
\textbf{\scriptsize FN} \\
\hline \hline
\textcolor{black}{SUP} & 471/471 & 0.61 & 461 & 114 & 357 & 10 \\
ATS & 471/471 & 0.86 & 394 & 414 & 57 & 77 \\
ATS 70/30 & 471/202 & 0.84 & 176 & 391 & 80 & 26 \\
ATS 90/10 & 471/52 & 0.81 & 43 & 382 & 89 & 9 \\
\hline \hline
\end{tabu}
\end{center}
\caption{Classification results of different experiments \textcolor{black}{performed with supervised classification (``SUP'') and the proposed method (``ATS'')} using the RANK database embedded with HUGO and $0.40$ bpp: \textcolor{black}{Accuracy (``Acc''), True positives (``TP''), True negatives (``TN''), False positives (``FP'') and False negatives (``FN'')}}
\label{tab:hugo40_bossrank}
\end{table}

\newpage

\begin{table}[ht]
\begin{center}
\resizebox{\textwidth}{!}{
\begin{tabu}{l|ccccc|ccccc|ccccc|ccccc}
\hhline{=====================} 
\multicolumn{21}{c}{LSB Matching Steganography with 50 stego images} \\ 
\hhline{---------------------} 
\multicolumn{1}{c}{} & 
\multicolumn{5}{|c}{\textbf{SUP $\mathbf{0.25}$ bpp}} & 
\multicolumn{5}{|c}{\textbf{ATS $\mathbf{0.25}$ bpp}} &
\multicolumn{5}{|c}{\textbf{SUP $\mathbf{0.10}$ bpp}} & 
\multicolumn{5}{|c}{\textbf{ATS $\mathbf{0.10}$ bpp}} \\
\hhline{-|-----|-----|-----|-----} 
\textbf{DB} & 
\textbf{\scriptsize Acc} & 
\textbf{\scriptsize TP} & 
\textbf{\scriptsize TN} & 
\textbf{\scriptsize FP} & 
\textbf{\scriptsize FN} & 
\textbf{\scriptsize Acc } &
\textbf{\scriptsize TP} & 
\textbf{\scriptsize TN} & 
\textbf{\scriptsize FP} & 
\textbf{\scriptsize FN} &
\textbf{\scriptsize Acc} & 
\textbf{\scriptsize TP} & 
\textbf{\scriptsize TN} & 
\textbf{\scriptsize FP} & 
\textbf{\scriptsize FN} & 
\textbf{\scriptsize Acc } &
\textbf{\scriptsize TP} & 
\textbf{\scriptsize TN} & 
\textbf{\scriptsize FP} & 
\textbf{\scriptsize FN} \\
\hline \hline

BOSS & 
\textbf{0.95} & 47 & 119 & 6 & 3 & 
0.85 & 48 & 100 & 25 & 2 & 
\textbf{0.89} & 47 & 108 & 17 & 3 & 
0.78 & 48 & 89 & 36 & 2 \\

ESO & 
0.47 & 12 & 71 & 54 & 38 & 
\textbf{0.82} & 47 & 96 & 29 & 3 & 
0.30 & 45 & 8 & 117 & 5 & 
\textbf{0.89} & 48 & 108 & 17 & 2 \\

CALP & 
0.72 & 4 & 117 & 0 & 46 & 
\textbf{0.94} & 43 & 121 & 4 & 7 & 
0.72 & 16 & 105 & 12 & 34 & 
\textbf{0.85} & 37 & 112 & 13 & 13 \\

INTE & 
0.66 & 2 & 114 & 11 & 48 & 
\textbf{0.97} & 49 & 120 & 5 & 1 & 
0.29 & 50 & 0 & 125 & 0 & 
\textbf{0.92} & 49 & 112 & 13 & 1 \\

NRCS & 
\textbf{0.66} & 5 & 103 & 11 & 45 & 
0.59 & 41 & 63 & 62 & 9 & 
0.39 & 40 & 24 & 90 & 10 & 
\textbf{0.52} & 37 & 54 & 71 & 13 \\

ALBI & 
0.61 & 22 & 84 & 41 & 28 & 
\textbf{0.98} & 50 & 122 & 3 & 0 & 
0.55 & 42 & 55 & 70 & 8 & 
\textbf{0.97} & 49 & 121 & 4 & 1 \\

NOAA & 
0.46 & 17 & 64 & 61 & 33 & 
\textbf{0.89} & 50 & 105 & 20 & 0 & 
0.29 & 50 & 0 & 125 & 0 & 
\textbf{0.95} & 49 & 117 & 8 & 1 \\

\hline \hline
\end{tabu}
}
\end{center}
\caption{Classification results of the supervised approach \textcolor{black}{(``SUP'')} and the proposed method \textcolor{black}{(``ATS'')} \textcolor{black}{for different databases (``DB'')} embedded with LSB matching steganography using only 50 stego images: \textcolor{black}{Accuracy (``Acc''), True positives (``TP''), True negatives (``TN''), False positives (``FP'') and False negatives (``FN'')}}
\label{tab:results_LSBM_50}
\end{table}

\begin{table}[ht]
\begin{center}
\resizebox{\textwidth}{!}{
\begin{tabu}{l|ccccc|ccccc|ccccc|ccccc}
\hhline{=====================} 
\multicolumn{21}{c}{HUGO Steganography with 50 stego images} \\ 
\hhline{---------------------} 
\multicolumn{1}{c}{} & 
\multicolumn{5}{|c}{\textbf{SUP $\mathbf{0.40}$ bpp}} & 
\multicolumn{5}{|c}{\textbf{ATS $\mathbf{0.40}$ bpp}} &
\multicolumn{5}{|c}{\textbf{SUP $\mathbf{0.20}$ bpp}} & 
\multicolumn{5}{|c}{\textbf{ATS $\mathbf{0.20}$ bpp}} \\
\hhline{-|-----|-----|-----|-----} 
\textbf{DB} & 
\textbf{\scriptsize Acc} & 
\textbf{\scriptsize TP} & 
\textbf{\scriptsize TN} & 
\textbf{\scriptsize FP} & 
\textbf{\scriptsize FN} & 
\textbf{\scriptsize Acc } &
\textbf{\scriptsize TP} & 
\textbf{\scriptsize TN} & 
\textbf{\scriptsize FP} & 
\textbf{\scriptsize FN} &
\textbf{\scriptsize Acc} & 
\textbf{\scriptsize TP} & 
\textbf{\scriptsize TN} & 
\textbf{\scriptsize FN} & 
\textbf{\scriptsize FP} & 
\textbf{\scriptsize Acc } &
\textbf{\scriptsize TP} & 
\textbf{\scriptsize TN} & 
\textbf{\scriptsize FP} & 
\textbf{\scriptsize FN} \\
\hline \hline

BOSS & 
0.77 & 43 & 91 & 34 & 7 & 
\textbf{0.79} & 46 & 92 & 33 & 4  & 
\textbf{0.63} & 36 & 75 & 50 & 14 & 
0.62 & 40 & 69 & 56 & 10 \\

ESO & 
0.43 & 34 & 41 & 84 & 16 & 
\textbf{0.82} & 49 & 95 & 30 & 1 & 
0.33 & 34 & 23 & 102 & 16 & 
\textbf{0.76} & 47 & 86 & 39 & 3 \\

CALP & 
0.44 & 21 & 51 & 66 & 24 & 
\textbf{0.77} & 36 & 89 & 28 & 9 & 
0.49 & 13 & 66 & 51 & 32 & 
\textbf{0.69} & 37 & 75 & 42 & 8 \\

INTE & 
0.61 & 22 & 85 & 40 & 28 & 
\textbf{0.94} & 50 & 114 & 11 & 0 & 
0.43 & 27 & 48 & 77 & 23 & 
\textbf{0.89} & 49 & 106 & 19 & 1 \\

NRCS & 
0.35 & 29 & 25 & 89 & 11 & 
\textbf{0.68} & 36 & 68 & 46 & 4  & 
0.49 & 25 & 51 & 63 & 15 & 
\textbf{0.72} & 7 & 104 & 10 & 33 \\

ALBI & 
0.30 & 48 & 5 & 120 & 2 & 
\textbf{0.95} & 49 & 117 & 8 & 1 & 
0.32 & 46 & 10 & 115 & 4 & 
\textbf{0.74} & 42 & 87 & 38 & 8 \\

NOAA & 
0.31 & 49 & 6 & 119 & 1 & 
\textbf{0.98} & 49 & 123 & 2 & 1  & 
0.30 & 49 & 3 & 122 & 1 & 
\textbf{0.90} & 49 & 109 & 16 & 1 \\

\hline \hline
\end{tabu}
}
\end{center}
\caption{Classification results of the supervised approach \textcolor{black}{(``SUP'')} and the proposed method \textcolor{black}{(``ATS'')} \textcolor{black}{for different databases (``DB'')} embedded with HUGO steganography using only 50 stego images: \textcolor{black}{Accuracy (``Acc''), True positives (``TP''), True negatives (``TN''), False positives (``FP'') and False negatives (``FN'')}}
\label{tab:results_HUGO_50}
\end{table}

\newpage

\begin{table}[ht]
\begin{center}
\resizebox{\textwidth}{!}{
\begin{tabu}{l|ccccc|ccccc|ccccc|ccccc}
\hhline{=====================} 
\multicolumn{21}{c}{LSB Matching Steganography with 10 stego images} \\ 
\hhline{---------------------} 
\multicolumn{1}{c}{} & 
\multicolumn{5}{|c}{\textbf{SUP $\mathbf{0.25}$ bpp}} & 
\multicolumn{5}{|c}{\textbf{ATS $\mathbf{0.25}$ bpp}} &
\multicolumn{5}{|c}{\textbf{SUP $\mathbf{0.10}$ bpp}} & 
\multicolumn{5}{|c}{\textbf{ATS $\mathbf{0.10}$ bpp}} \\
\hhline{-|-----|-----|-----|-----} 
\textbf{DB} & 
\textbf{\scriptsize Acc} & 
\textbf{\scriptsize TP} & 
\textbf{\scriptsize TN} & 
\textbf{\scriptsize FP} & 
\textbf{\scriptsize FN} & 
\textbf{\scriptsize Acc } &
\textbf{\scriptsize TP} & 
\textbf{\scriptsize TN} & 
\textbf{\scriptsize FP} & 
\textbf{\scriptsize FN} &
\textbf{\scriptsize Acc} & 
\textbf{\scriptsize TP} & 
\textbf{\scriptsize TN} & 
\textbf{\scriptsize FP} & 
\textbf{\scriptsize FN} & 
\textbf{\scriptsize Acc } &
\textbf{\scriptsize TP} & 
\textbf{\scriptsize TN} & 
\textbf{\scriptsize FP} & 
\textbf{\scriptsize FN} \\
\hline \hline

BOSS & 
\textbf{0.96} & 10 & 119 & 6 & 0 & 
0.86 & 10 & 106 & 19 & 0 & 
\textbf{0.87} & 9 & 108 & 17 & 1 & 
0.70 & 9 & 86 & 39 & 1 \\

ESO & 
0.54 & 2 & 71 & 54 & 8 & 
\textbf{0.82} & 10 & 101 & 24 & 0 & 
0.13 & 9 & 8 & 117 & 1 & 
\textbf{0.58} & 10 & 68 & 57 & 0  \\

CALP & 
\textbf{0.92} & 0 & 117 & 0 & 10 & 
0.29 & 6 & 33 & 92 & 4 & 
\textbf{0.87} & 5 & 105 & 12 & 5 & 
0.26 & 4 & 31 & 94 & 6 \\

NTE & 
\textbf{0.84} & 0 & 114 & 11 & 10 & 
0.64 & 8 & 78 & 47 & 2 & 
0.07 & 10 & 0 & 125 & 0 & 
\textbf{0.62} & 10 & 74 & 51 & 0 \\

NRCS & 
\textbf{0.83} & 0 & 103 & 11 & 10 & 
0.47 & 9 & 54 & 71 & 1 & 
0.26 & 8 & 24 & 90 & 2 & 
\textbf{0.41} & 8 & 48 & 77 & 2 \\

ALBI & 
0.69 & 9 & 84 & 41 & 1 & 
\textbf{0.91} & 10 & 113 & 12 & 0 & 
0.47 & 8 & 55 & 70 & 2 & 
\textbf{0.81} & 10 & 100 & 25 & 0 \\

NOAA & 
0.48 & 1 & 64 & 61 & 9 & 
\textbf{0.72} & 9 & 88 & 37 & 1 & 
0.07 & 10 & 0 & 125 & 0 & 
\textbf{0.73} & 10 & 88 & 37 & 0 \\

\hline \hline
\end{tabu}
}
\end{center}
\caption{Classification results of the suppervised approach \textcolor{black}{(``SUP'')} and the proposed method \textcolor{black}{(``ATS'')} \textcolor{black}{for different databases (``DB'')} embedded with LSB matching steganography using only 10 stego images: \textcolor{black}{Accuracy (``Acc''), True positives (``TP''), True negatives (``TN''), False positives (``FP'') and False negatives (``FN'')}}
\label{tab:results_LSBM_10}
\end{table}

\begin{table}[ht]
\begin{center}
\resizebox{\textwidth}{!}{
\begin{tabu}{l|ccccc|ccccc|ccccc|ccccc}
\hhline{=====================} 
\multicolumn{21}{c}{HUGO Steganography with 10 stego images} \\ 
\hhline{---------------------} 
\multicolumn{1}{c}{} & 
\multicolumn{5}{|c}{\textbf{SUP $\mathbf{0.40}$ bpp}} & 
\multicolumn{5}{|c}{\textbf{ATS $\mathbf{0.40}$ bpp}} &
\multicolumn{5}{|c}{\textbf{SUP $\mathbf{0.20}$ bpp}} & 
\multicolumn{5}{|c}{\textbf{ATS $\mathbf{0.20}$ bpp}} \\
\hhline{-|-----|-----|-----|-----} 
\textbf{DB} & 
\textbf{\scriptsize Acc} & 
\textbf{\scriptsize TP} & 
\textbf{\scriptsize TN} & 
\textbf{\scriptsize FP} & 
\textbf{\scriptsize FN} & 
\textbf{\scriptsize Acc } &
\textbf{\scriptsize TP} & 
\textbf{\scriptsize TN} & 
\textbf{\scriptsize FP} & 
\textbf{\scriptsize FN} &
\textbf{\scriptsize Acc} & 
\textbf{\scriptsize TP} & 
\textbf{\scriptsize TN} & 
\textbf{\scriptsize FP} & 
\textbf{\scriptsize FN} & 
\textbf{\scriptsize Acc } &
\textbf{\scriptsize TP} & 
\textbf{\scriptsize TN} & 
\textbf{\scriptsize FP} & 
\textbf{\scriptsize FN} \\
\hline \hline

BOSS & 
\textbf{0.75} & 10 & 91 & 34 & 0 & 
0.58 & 10 & 68 & 57 & 0 & 
\textbf{0.60} & 6 & 75 & 50 & 4 & 
0.51 & 10 & 59 & 66 & 0 \\

ESO & 
\textbf{0.36} & 7 & 41 & 84 & 3 & 
0.21 & 10 & 19 & 106 & 0 & 
0.22 & 7 & 23 & 102 & 3 & 
\textbf{0.43} & 9 & 49 & 76 & 1 \\

CALP & 
\textbf{0.40} & 0 & 51 & 66 & 9 & 
0.19 & 4 & 20 & 97 & 5 & 
\textbf{0.52} & 0 & 66 & 51 & 9 & 
0.16 & 3 & 17 & 100 & 6 \\

INTE & 
\textbf{0.64} & 2 & 85 & 40 & 8 & 
0.24 & 10 & 23 & 102 & 0 & 
\textbf{0.38} & 3 & 48 & 77 & 7 & 
0.27 & 10 & 27 & 98 & 0 \\

NRCS & 
0.26 & 7 & 25 & 89 & 3 & 
\textbf{0.51} & 8 & 55 & 59 & 2 & 
0.45 & 5 & 51 & 63 & 5 & 
\textbf{0.49} & 8 & 53 & 61 & 2 \\

ALBI & 
0.11 & 10 & 5 & 120 & 0 & 
\textbf{0.65} & 10 & 78 & 47 & 0 & 
0.15 & 10 & 10 & 115 & 0 & 
\textbf{0.57} & 9 & 68 & 57 & 1 \\

NOAA & 
0.11 & 9 & 6 & 119 & 1 & 
\textbf{0.60} & 10 & 71 & 54 & 0 & 
0.09 & 9 & 3 & 122 & 1 & 
\textbf{0.76} & 10 & 93 & 32 & 0 \\

\hline \hline
\end{tabu}
}
\end{center}
\caption{Classification results of the supervised approach \textcolor{black}{(``SUP'')} and the proposed method \textcolor{black}{(``ATS'')} \textcolor{black}{for different databases (``DB'')} embedded with HUGO steganography using only 10 stego images: \textcolor{black}{Accuracy (``Acc''), True positives (``TP''), True negatives (``TN''), False positives (``FP'') and False negatives (``FN'')}}
\label{tab:results_HUGO_10}
\end{table}

\newpage

\begin{table}[ht]
\begin{center}
\begin{tabu}{ccccccc}
\hhline{=======} 
\multicolumn{7}{c}{Variable ratio of stego images} \\ 
\hhline{-------} 
\multicolumn{7}{c}{\textbf{HUGO $\mathbf{0.40}$ bpp}} \\
\hhline{-------} 
\textbf{\scriptsize Cover/Stego} & 
\textbf{\scriptsize Percent of stego} & 
\textbf{\scriptsize Acc} & 
\textbf{\scriptsize TP} & 
\textbf{\scriptsize TN} & 
\textbf{\scriptsize FP} & 
\textbf{\scriptsize FN} \\
\hline \hline

125 / 0  & 0\%  & 0.26 &  0 & 33 & 92 & 0 \\
120 / 5  & 4\%  & 0.54 &  5 & 63 & 57 & 0 \\      
115 / 10 & 8\%  & 0.53 & 10 & 56 & 59 & 0 \\      
110 / 15 & 12\% & 0.73 & 15 & 76 & 34 & 0 \\      
105 / 20 & 16\% & 0.73 & 19 & 72 & 33 & 1 \\      
100 / 25 & 20\% & 0.76 & 23 & 72 & 28 & 2 \\        
95 / 30  & 24\% & 0.81 & 28 & 73 & 22 & 2 \\        
90 / 35  & 28\% & 0.82 & 31 & 72 & 18 & 4 \\      
85 / 40  & 32\% & 0.80 & 37 & 63 & 22 & 3  \\      
80 / 45  & 36\% & 0.82 & 42 & 61 & 19 & 3 \\      
75 / 50  & 40\% & 0.81 & 44 & 57 & 18 & 6 \\      
70 / 55  & 44\% & 0.87 & 52 & 57 & 13 & 3 \\      
65 / 60  & 48\% & 0.81 & 51 & 50 & 15 & 9 \\      
60 / 65  & 52\% & 0.85 & 58 & 48 & 12 & 7 \\      
55 / 70  & 56\% & 0.84 & 62 & 43 & 12 & 8 \\      
50 / 75  & 60\% & 0.85 & 61 & 45 & 5 & 14 \\      
45 / 80  & 64\% & 0.85 & 67 & 39 & 6 & 13 \\      
40 / 85  & 68\% & 0.84 & 71 & 34 & 6 & 14 \\      
35 / 90  & 72\% & 0.86 & 74 & 33 & 2 & 16 \\      
30 / 95  & 76\% & 0.85 & 80 & 26 & 4 & 15 \\      
25 / 100 & 80\% & 0.82 & 82 & 20 & 5 & 18 \\       
20 / 105 & 84\% & 0.80 & 80 & 20 & 0 & 25 \\   
15 / 110 & 88\% & 0.79 & 85 & 14 & 1 & 25 \\
10 / 115 & 92\% & 0.75 & 86 & 8 & 2 & 29 \\
5 / 120  & 96\% & 0.75 & 89 & 5 & 0 & 31 \\
0 / 125  & 100\% & 0.65 & 81 & 0 & 0 & 44  \\
\hline \hline
\end{tabu}
\end{center}
\caption{Classification results of the proposed method with the BOSS database \textcolor{black}{(``DB'')} embedded with HUGO and $0.40$ bpp: \textcolor{black}{Accuracy (``Acc''), True positives (``TP''), True negatives (``TN''), False positives (``FP'') and False negatives (``FN'')}}
\label{tab:hugo40_step5}
\end{table}

\begin{table}[ht]
\begin{center}
\begin{tabu}{cccccccc}
\hhline{========} 
\multicolumn{8}{c}{HUGO/LSBM with few testing images} \\

\textbf{\scriptsize Cover/Stego} & 
\textbf{\scriptsize DB/Algorithm} &
\textbf{\scriptsize Acc} & 
\textbf{\scriptsize TP} & 
\textbf{\scriptsize TN} & 
\textbf{\scriptsize FP} & 
\textbf{\scriptsize FN} \\
\hline \hline

10/10 & \scriptsize BOSS / LSBM 0.25 & 0.89 & 9.1 & 8.8 & 1.2 & 0.9 \\
10/10 & \scriptsize BOSS / HUGO 0.40 & 0.73 & 7.8 & 6.8 & 3.2 & 2.2 \\
10/10 & \scriptsize ESO  / LSBM 0.25 & 0.82 & 8.6 & 7.9 & 2.1 & 1.4 \\ 
10/10 & \scriptsize ESO  / HUGO 0.40 & 0.68 & 7.8 & 5.8 & 4.2 & 2.2 \\
10/10 & \scriptsize CALP / LSBM 0.25 & 0.84 & 8.5 & 8.3 & 1.7 & 1.5 \\          
                                          
10/10 & \scriptsize CALP / HUGO 0.40 & 0.80 & 7.7 & 7.4 & 2.2 & 1.5 \\          
                                            
10/10 & \scriptsize INTE / LSBM 0.25 & 0.94 & 10 & 8.8 & 1.2 & 0 \\           
                                            
10/10 & \scriptsize INTE / HUGO 0.40 & 0.77 & 9.3 & 6.1 & 3.9 & 0.7 \\          
                                             
10/10 & \scriptsize NRCS / LSBM 0.25 & 0.53 & 6.3 & 4.3 & 5.7 & 3.7 \\          
                                             
10/10 & \scriptsize NRCS / HUGO 0.40 & 0.53 & 4.7 & 5.2 & 4.3 & 4.2 \\          
                                           
10/10 & \scriptsize ALBN / LSBM 0.25 & 0.93 & 9.4 & 9.2 & 0.8 & 0.6 \\ 
10/10 & \scriptsize ALBN / HUGO 0.40 & 0.85 & 9.2 & 7.9 & 2.1 & 0.8 \\  
10/10 & \scriptsize NOAA / LSBM 0.25 & 0.93 & 9.6 & 9.1 & 0.9 & 0.4 \\ 
10/10 & \scriptsize NOAA / HUGO 0.40 & 0.83 & 9.2 & 7.5 & 2.5 & 0.8 \\

\hhline{--------}                                                   

5/5 & \scriptsize BOSS / LSBM 0.25 & 0.83 & 4.0 & 4.3 & 0.7 & 1.0 \\
5/5 & \scriptsize BOSS / HUGO 0.40 & 0.61 & 3.2 & 2.9 & 2.1 & 1.8 \\
5/5 & \scriptsize ESO  / LSBM 0.25 & 0.71 & 4.0 & 3.1 & 1.9 & 1.0 \\
5/5 & \scriptsize ESO  / HUGO 0.40 & 0.68 & 3.9 & 2.9 & 2.1 & 1.1 \\
5/5 & \scriptsize CALP / LSBM 0.25 & 0.79 & 3.7 & 4.2 & 0.8 & 1.3 \\
5/5 & \scriptsize CALP / HUGO 0.40 & 0.63 & 3.6 & 2.3 & 2.4 & 1.0 \\
5/5 & \scriptsize INTE / LSBM 0.25 & 0.86 & 4.9 & 3.7 & 1.3 & 0.1 \\
5/5 & \scriptsize INTE / HUGO 0.40 & 0.73 & 4.4 & 2.9 & 2.1 & 0.6 \\
5/5 & \scriptsize NRCS / LSBM 0.25 & 0.51 & 3.5 & 1.6 & 3.4 & 1.5 \\
5/5 & \scriptsize NRCS / HUGO 0.40 & 0.59 & 2.7 & 2.6 & 2.0 & 1.7 \\
5/5 & \scriptsize ALBN / LSBM 0.25 & 0.85 & 4.3 & 4.2 & 0.8 & 0.7 \\
5/5 & \scriptsize ALBN / HUGO 0.40 & 0.77 & 4.1 & 3.6 & 1.4 & 0.9 \\
5/5 & \scriptsize NOAA / LSBM 0.25 & 0.82 & 4.5 & 3.7 & 1.3 & 0.5 \\
5/5 & \scriptsize NOAA / HUGO 0.40 & 0.78 & 4.6 & 3.2 & 1.8 & 0.4 \\

\hline \hline
\end{tabu}

\end{center}
\caption{Average classification results of the proposed method for different databases \textcolor{black}{(``DB'')} with HUGO and LSB matching steganography for very small testing sets: \textcolor{black}{Accuracy (``Acc''), True positives (``TP''), True negatives (``TN''), False positives (``FP'') and False negatives (``FN'')}}
\label{tab:results_few1}

\end{table}

\begin{table}[ht]
\begin{center}

\begin{tabu}{cccccccc}
\hhline{========} 
\multicolumn{8}{c}{HUGO/LSBM with few testing images} \\

\textbf{\scriptsize Cover/Stego} & 
\textbf{\scriptsize DB/Algorithm} &
\textbf{\scriptsize Acc} & 
\textbf{\scriptsize TP} & 
\textbf{\scriptsize TN} & 
\textbf{\scriptsize FP} & 
\textbf{\scriptsize FN} \\
\hline \hline

5/10 & \scriptsize BOSS / LSBM 0.25 & 0.77 & 7.8 & 3.8 & 1.2 & 2.2 \\  

5/10 & \scriptsize BOSS / HUGO 0.40 & 0.69 & 7.4 & 3.0 & 2.0 & 2.6 \\   

5/10 & \scriptsize ESO  / LSBM 0.25 & 0.80 & 8.7 & 3.4 & 1.6 & 1.3 \\   

5/10 & \scriptsize ESO  / HUGO 0.40 & 0.77 & 7.8 & 3.8 & 1.2 & 2.2 \\  

5/10 & \scriptsize CALP / LSBM 0.25 & 0.77 & 7.8 & 3.9 & 1.1 & 2.2 \\  

5/10 & \scriptsize CALP / HUGO 0.40 & 0.83 & 8.1 & 3.8 & 0.7 & 1.7 \\  

5/10 & \scriptsize INTE / LSBM 0.25 & 0.93 & 9.6 & 4.4 & 0.6 & 0.4 \\  

5/10 & \scriptsize INTE / HUGO 0.40 & 0.75 & 7.6 & 3.7 & 1.3 & 2.4 \\  

5/10 & \scriptsize NRCS / LSBM 0.25 & 0.61 & 7.0 & 2.3 & 2.7 & 3.0 \\   

5/10 & \scriptsize NRCS / HUGO 0.40 & 0.49 & 4.1 & 2.4 & 2.3 & 4.5 \\  

5/10 & \scriptsize ALBN / LSBM 0.25 & 0.92 & 9.5 & 4.4 & 0.6 & 0.5 \\  

5/10 & \scriptsize ALBN / HUGO 0.40 & 0.86 & 8.6 & 4.3 & 0.7 & 1.4 \\   

5/10 & \scriptsize NOAA / LSBM 0.25 & 0.94 & 9.6 & 4.5 & 0.5 & 0.4 \\   

5/10 & \scriptsize NOAA / HUGO 0.40 & 0.89 & 8.9 & 4.5 & 0.5 & 1.1 \\ 

\hhline{--------}

10/5 & \scriptsize BOSS / LSBM 0.25 & 0.81 & 4.5 & 7.7 & 2.3 & 0.5 \\   

10/5 & \scriptsize BOSS / HUGO 0.40 & 0.59 & 3.4 & 5.5 & 4.5 & 1.6 \\  

10/5 & \scriptsize ESO  / LSBM 0.25 & 0.73 & 4.7 & 6.3 & 3.7 & 0.3 \\   
10/5 & \scriptsize ESO  / HUGO 0.40 & 0.66 & 4.2 & 5.8 & 4.2 & 0.8 \\   
10/5 & \scriptsize CALP / LSBM 0.25 & 0.82 & 3.9 & 8.4 & 1.6 & 1.1 \\   
10/5 & \scriptsize CALP / HUGO 0.40 & 0.57 & 3.5 & 4.7 & 4.6 & 1.5 \\   
10/5 & \scriptsize INTE / LSBM 0.25 & 0.87 & 4.9 & 8.2 & 1.8 & 0.1 \\   
10/5 & \scriptsize INTE / HUGO 0.40 & 0.69 & 4.9 & 5.6 & 4.4 & 0.1 \\   
10/5 & \scriptsize NRCS / LSBM 0.25 & 0.48 & 3.6 & 3.6 & 6.4 & 1.4 \\  
10/5 & \scriptsize NRCS / HUGO 0.40 & 0.52 & 2.8 & 4.3 & 5.0 & 1.6 \\  
10/5 & \scriptsize ALBN / LSBM 0.25 & 0.88 & 4.7 & 8.6 & 1.4 & 0.3 \\   
10/5 & \scriptsize ALBN / HUGO 0.40 & 0.76 & 4.7 & 6.7 & 3.3 & 0.3 \\  
10/5 & \scriptsize NOAA / LSBM 0.25 & 0.88 & 4.8 & 8.5 & 1.5 & 0.2 \\   
10/5 & \scriptsize NOAA / HUGO 0.40 & 0.71 & 4.7 & 6.1 & 3.9 & 0.3 \\

\hline \hline
\end{tabu}

\end{center}
\caption{Average classification results of the proposed method for different databases \textcolor{black}{(``DB'')} with HUGO and LSB matching steganography for very small testing sets: \textcolor{black}{Accuracy (``Acc''), True positives (``TP''), True negatives (``TN''), False positives (``FP'') and False negatives (``FN'')}}
\label{tab:results_few2}

\end{table}

\begin{table}[ht]
\begin{center}

\begin{tabu}{cccccccc}
\hhline{========} 
\multicolumn{8}{c}{HUGO/LSBM with few testing images} \\

\textbf{\scriptsize Cover/Stego} & 
\textbf{\scriptsize DB/Algorithm} &
\textbf{\scriptsize Acc} & 
\textbf{\scriptsize TP} & 
\textbf{\scriptsize TN} & 
\textbf{\scriptsize FP} & 
\textbf{\scriptsize FN} \\
\hline \hline

5/3 & \scriptsize BOSS / LSBM 0.25 & 0.75 & 2.5 & 3.5 & 1.5 & 0.5 \\
5/3 & \scriptsize BOSS / HUGO 0.40 & 0.55 & 2.1 & 2.3 & 2.7 & 0.9 \\
5/3 & \scriptsize ESO  / LSBM 0.25 & 0.64 & 2.7 & 2.4 & 2.6 & 0.3 \\
5/3 & \scriptsize ESO  / HUGO 0.40 & 0.59 & 2.4 & 2.3 & 2.7 & 0.6 \\
5/3 & \scriptsize CALP / LSBM 0.25 & 0.75 & 1.9 & 4.1 & 0.9 & 1.1 \\
5/3 & \scriptsize CALP / HUGO 0.40 & 0.61 & 2.4 & 2.3 & 2.4 & 0.6 \\
5/3 & \scriptsize INTE / LSBM 0.25 & 0.77 & 3.0 & 3.2 & 1.8 & 0 \\
5/3 & \scriptsize INTE / HUGO 0.40 & 0.71 & 2.9 & 2.8 & 2.2 & 0.1 \\
5/3 & \scriptsize NRCS / LSBM 0.25 & 0.45 & 2.1 & 1.5 & 3.5 & 0.9 \\
5/3 & \scriptsize NRCS / HUGO 0.40 & 0.49 & 1.3 & 2.3 & 2.3 & 1.2 \\
5/3 & \scriptsize ALBN / LSBM 0.25 & 0.65 & 2.3 & 2.9 & 2.1 & 0.7 \\
5/3 & \scriptsize ALBN / HUGO 0.40 & 0.64 & 2.5 & 2.6 & 2.4 & 0.5 \\
5/3 & \scriptsize NOAA / LSBM 0.25 & 0.79 & 2.9 & 3.4 & 1.6 & 0.1 \\
5/3 & \scriptsize NOAA / HUGO 0.40 & 0.61 & 3.0 & 1.9 & 3.1 & 0 \\ 
\hhline{--------} 
3/5 & \scriptsize BOSS / LSBM 0.25 & 0.81 & 4.3 & 2.2 & 0.8 & 0.7 \\
3/5 & \scriptsize BOSS / HUGO 0.40 & 0.58 & 3.1 & 1.6 & 1.4 & 1.9 \\
3/5 & \scriptsize ESO  / LSBM 0.25 & 0.67 & 3.8 & 1.6 & 1.4 & 1.2 \\
3/5 & \scriptsize ESO  / HUGO 0.40 & 0.74 & 4.5 & 1.4 & 1.6 & 0.5 \\
3/5 & \scriptsize CALP / LSBM 0.25 & 0.69 & 3.5 & 2.0 & 1.0 & 1.5 \\
3/5 & \scriptsize CALP / HUGO 0.40 & 0.61 & 2.9 & 1.4 & 1.1 & 1.6 \\
3/5 & \scriptsize INTE / LSBM 0.25 & 0.88 & 5.0 & 2.1 & 0.9 & 0 \\
3/5 & \scriptsize INTE / HUGO 0.40 & 0.84 & 4.6 & 2.1 & 0.9 & 0.4 \\
3/5 & \scriptsize NRCS / LSBM 0.25 & 0.54 & 3.5 & 0.8 & 2.2 & 1.5 \\
3/5 & \scriptsize NRCS / HUGO 0.40 & 0.54 & 2.7 & 1.4 & 1.6 & 1.9 \\
3/5 & \scriptsize ALBN / LSBM 0.25 & 0.64 & 3.5 & 1.6 & 1.4 & 1.5 \\
3/5 & \scriptsize ALBN / HUGO 0.40 & 0.79 & 3.9 & 2.4 & 0.6 & 1.1 \\
3/5 & \scriptsize NOAA / LSBM 0.25 & 0.86 & 4.8 & 2.1 & 0.9 & 0.2 \\    
3/5 & \scriptsize NOAA / HUGO 0.40 & 0.65 & 3.4 & 1.8 & 1.2 & 1.6 \\

\hline \hline
\end{tabu}

\end{center}
\caption{Average classification results of the proposed method for different databases \textcolor{black}{(``DB'')} with HUGO and LSB matching steganography for very small testing sets: \textcolor{black}{Accuracy (``Acc''), True positives (``TP''), True negatives (``TN''), False positives (``FP'') and False negatives (``FN'')}}
\label{tab:results_few3}
\end{table}

\begin{table}[ht]
\begin{center}
\begin{tabu}{c|ccccc|ccccc}
\hhline{=|=====|=====} 
\multicolumn{11}{c}{LSB Matching Embedding with mixed DBs ($ \mathbf{0.10} $ bpp)} \\ 
\hhline{-|-----|-----} 
\multicolumn{1}{c}{} &
\multicolumn{5}{|c}{\textbf{SUP}} &
\multicolumn{5}{|c}{\textbf{ATS}} \\
\hhline{-|-----|-----} 
\textbf{\scriptsize Stego/Total samples} & 
\textbf{\scriptsize Acc} & 
\textbf{\scriptsize TP} & 
\textbf{\scriptsize TN} & 
\textbf{\scriptsize FP} & 
\textbf{\scriptsize FN} &
\textbf{\scriptsize Acc} & 
\textbf{\scriptsize TP} & 
\textbf{\scriptsize TN} & 
\textbf{\scriptsize FP} & 
\textbf{\scriptsize FN} \\
\hline \hline

140/280 & 0.47 & 78 & 54 & 84 & 62 & 
          \textbf{0.79} & 116 & 106 & 34 & 24 \\

50/190  & 0.42 & 25 & 54 & 84 & 25 & 
          \textbf{0.74} & 50 & 91 & 49 & 0\\

140/280 (same size) & 
          0.69 & 122 & 72 & 67 & 18 & 
          \textbf{0.87} & 122 & 122 & 18 & 18 \\

50/190 (same size)  & 
          0.61 & 44 & 72 & 67 & 6 & 
          \textbf{0.78} & 46 & 103 & 37 & 4\\

\hline \hline
\end{tabu}
\end{center}
\caption{Classification results of the supervised approach \textcolor{black}{(``SUP'')} and the proposed method \textcolor{black}{(``ATS'')} in an unevenly mixed database \textcolor{black}{(``DB'')} embedded with LSB matching and $0.10$ bpp: \textcolor{black}{Accuracy (``Acc''), True positives (``TP''), True negatives (``TN''), False positives (``FP'') and False negatives (``FN'')}}
\label{tab:results_LSBM_10_MIX}
\end{table}

\begin{table}[ht]
\begin{center}
\begin{tabu}{l|ccccc|ccccc}
\hhline{=|=====|=====} 
\multicolumn{11}{c}{LSB Matching Steganography} \\
\multicolumn{11}{c}{with Multiple Embedding Bit Rates} \\
\hhline{-|-----|-----} 
\multicolumn{11}{c}{\textbf{$\mathbf{0.25}$, $\mathbf{0.20}$, $\mathbf{0.15}$, $\mathbf{0.10}$ and $\mathbf{0.05}$ bpp}} \\
\hhline{-|-----|-----} 
\multicolumn{1}{c}{} &
\multicolumn{5}{|c}{\textbf{SUP}} &
\multicolumn{5}{|c}{\textbf{ATS}} \\
\hhline{-|-----|-----} 
\textbf{\scriptsize DB} & 
\textbf{\scriptsize Acc} & 
\textbf{\scriptsize TP} & 
\textbf{\scriptsize TN} & 
\textbf{\scriptsize FP} & 
\textbf{\scriptsize FN} &
\textbf{\scriptsize Acc} & 
\textbf{\scriptsize TP} & 
\textbf{\scriptsize TN} & 
\textbf{\scriptsize FP} & 
\textbf{\scriptsize FN} \\
\hline \hline

BOSS & 
0.53 & 124 & 9 & 116 & 1 & 
\textbf{0.87} & 99 & 119 & 6 & 26 \\

CALP & 
0.53 & 81 & 47 & 70 & 44 &
\textbf{0.89} & 103 & 120 & 5 & 22 \\

INTE & 
0.47 & 117 & 0 & 125 & 8 &
\textbf{0.88} & 100 & 119 & 6 & 25 \\

\hline \hline
\end{tabu}
\end{center}
\caption{Classification results of the supervised approach \textcolor{black}{(``SUP'')} and proposed method \textcolor{black}{(``ATS'')} using multiple embedding bit rates \textcolor{black}{with three different databases (``DB''): Accuracy (``Acc''), True positives (``TP''), True negatives (``TN''), False positives (``FP'') and False negatives (``FN'')}}
\label{tab:results_LSBM_MBR}
\end{table}

\begin{table}[ht]
\begin{center}
\begin{tabu}{ccccc}
\hhline{=====} 
\multicolumn{5}{c}{Unknown Message Length} \\ 
\hhline{-----} 
\textbf{\scriptsize Algorithm} & 
\textbf{\scriptsize Real bit rate (bpp)} & 
\textbf{\scriptsize Tentative bit rate (bpp)} & 
\textbf{\scriptsize Score} & 
\textbf{\scriptsize Accuracy} \\
\hline \hline
\multirow{6}{*}{LSBM} & \multirow{6}{*}{$0.40$} & 0.10 &  0.66 & 0.72 \\
 &  & 0.20 &  0.52 & 0.84 \\
 &  & 0.30 &  0.25 & 0.89 \\
 &  & 0.40 &  \textbf{0.04} & 0.89 \\
 &  & 0.50 &  0.41 & 0.82 \\
 &  & 0.60 &  1.19 & 0.85 \\
\hhline{=====}
\multirow{6}{*}{HUGO} & \multirow{6}{*}{$0.40$} & 0.10 &  0.79 & 0.57 \\
 &  & 0.20 &  1.05 & 0.57 \\
 &  & 0.30 &  0.67 & 0.77 \\
 &  & 0.40 &  \textbf{0.24} & 0.81 \\
 &  & 0.50 &  0.37 & 0.76 \\
 &  & 0.60 &  0.48 & 0.63 \\
\hhline{=====}
\multirow{6}{*}{WOW} & \multirow{6}{*}{$0.40$} & 0.10  &  0.86 & 0.53 \\
 &  & 0.20  &  0.80 & 0.59 \\
 &  & 0.30  &  0.69 & 0.73 \\
 &  & 0.40  &  \textbf{0.40} & 0.76 \\
 &  & 0.50  &  0.45 & 0.71 \\
 &  & 0.60  &  0.58 & 0.62 \\
\hhline{=====}
\end{tabu}
\end{center}
\caption{Classification results of different experiments using the BOSS database embedded with LSBM, HUGO and WOW with an unknown message length}
\label{tab:boss_unknown}
\end{table}

\end{document}